\documentclass[english,13pt]{article}
\usepackage[T1]{fontenc}
\usepackage[latin1]{inputenc}
\usepackage{geometry}
\usepackage{float}
\usepackage{mathrsfs}
\usepackage{amsmath}
\usepackage{amssymb}
\usepackage{graphicx}
\usepackage{hyperref}
\usepackage[color=yellow]{todonotes}
\hypersetup{
    colorlinks=true,
    linkcolor=blue,
    filecolor=magenta,      
    urlcolor=cyan,
    citecolor = blue,
}

\usepackage{yhmath}
\makeatletter

\usepackage{bm}
\numberwithin{equation}{section}

\usepackage{algorithm}
 \usepackage{algorithmicx}
\floatstyle{ruled}
\newfloat{algorithm}{tbp}{loa}
\providecommand{\algorithmname}{Algorithm}
\floatname{algorithm}{\protect\algorithmname}

\usepackage{algorithm}
 \usepackage{algorithmicx}

\usepackage{amsthm}

\usepackage{mathrsfs}
\usepackage{amssymb}

\usepackage{amsfonts}

\usepackage{epsfig}

\usepackage{bm}

\usepackage{mathrsfs}

\usepackage{enumerate}

\@ifundefined{definecolor}{\@ifundefined{definecolor}
 {\@ifundefined{definecolor}
 {\usepackage{color}}{}
}{}
}{}

\usepackage{subfig}\usepackage[all]{xy}

\newtheorem{theorem}{Theorem}[section]
\newtheorem{lem}{Lemma}[section]

\newcounter{hypA}
\newenvironment{hypA}{\refstepcounter{hypA}\begin{itemize}
  \item[({\bf A\arabic{hypA}})]}{\end{itemize}}
\newcounter{hypB}

\newcounter{hypD}
\newenvironment{hypD}{\refstepcounter{hypD}\begin{itemize}
 \item[({\bf D\arabic{hypD}})]}{\end{itemize}}

\usepackage{babel}\date{}

\makeatother

\begin{document}

\begin{center}

{\Large \textbf{Bayesian Parameter Estimation for Partially Observed McKean-Vlasov Diffusions Using Multilevel Markov chain Monte Carlo}}

\vspace{0.5cm}

AJAY JASRA$^{1}$, \& AMIN WU$^{2}$

{\footnotesize $^{1}$School of Data Science,  The Chinese University of Hong Kong, Shenzhen,  Shenzhen,  CN.}\\
{\footnotesize $^{2}$Statistics Program, Computer, Electrical and Mathematical Sciences and Engineering Division, \\ King Abdullah University of Science and Technology, Thuwal, 23955-6900, KSA.} \\
{\footnotesize E-Mail:\,} \texttt{\emph{\footnotesize ajayjasra@cuhk.edu.cn}}, 
\texttt{\emph{\footnotesize amin.wu@kaust.edu.sa}}

\begin{abstract}
In this article we consider Bayesian estimation of static parameters for a class of partially observed McKean-Vlasov diffusion processes with discrete-time observations over a fixed time interval.  
This problem features several obstacles to its solution,  which include that the posterior density is numerically intractable in continuous-time,  even if the transition probabilities are available and even when one uses a time-discretization,  the posterior still cannot be used by adopting well-known computational methods such as Markov chain Monte Carlo (MCMC).  In this paper we provide a solution to this problem by using new MCMC algorithms which can solve the afore-mentioned issues.  This MCMC algorithm is extended to use multilevel Monte Carlo (MLMC) methods.  We prove convergence bounds on our parameter estimators and show that the MLMC-based MCMC algorithm reduces the computational cost to achieve a mean square error versus ordinary MCMC by an order of magnitude.  We numerically illustrate our results on two models.
\end{abstract}

\end{center}

\section{Introduction}

We consider the problem of Bayesian parameter estimation for partially observed (PO) McKean-Vlasov (MV) stochastic differential equtions (SDE). MV SDEs find a wide variety of applications in mathematics,  physics and beyond; see for instance \cite{bald,crisan,fin,biol} and the references therein.  We will assume that data are observed in discrete time and are conditionally independent given the position of the MV SDE at the observation time.  Given an overall fixed time window we seek to estimate the parameters of the model using Bayesian methods.  Several approaches for parameter inference or computation associated to MV models can be found in the works \cite{po_mv,nadir,mv_pe,nik}, but to the best of our knowledge there does not seem to be any work on the Bayesian estimation of parameters associated to a partially observed MV SDE.

The posterior density of the unknown parameters and the states of the MV SDE at observations is typically unavailable as the transition measure of the MV process is seldom available.  This often means one must resort to a time-discretization,  for instance based on the Euler-Maruyama method.  However,  unlike the problem with regular SDEs (see e.g.~\cite{chada_ub,jasra_bpe_sde}) even in such a case the posterior density of the parameters and states is unavailable up-to a normalizing constant or a non-negative unbiased estimator,  which is a prerequisite of using advanced computational tools to sample from a probability density (see e.g.~\cite{ml_rev} and the references therein).  The main problem is the fact that the laws of the MV diffusion,  which are a key component of the drift and diffusion coefficients,  are unknown even when time-discretizing the diffusion.  A solution to this,  in the context of data problems has been used in \cite{po_mv},  which is to input a particle approximation, based on the work of \cite{basic_method} for the said laws and to then to continue with the problem of interest,  given this additional bias; this is the line of thinking which we follow in this article.

In this paper we consider  Markov chain Monte Carlo (MCMC) algorithms to sample from the posterior density of the unknown parameters and states associated to PO MV SDE,  when using two level of approximations. The first is time-discretization and the second a particle approximation of the laws of the MV SDE.  This leads to a non-trivial probability density for which we design particle MCMC \cite{andrieu} algorithms,  which are related to those in \cite{mv_pe}.  The problem features an element of time discretization and it is well known in the literature that multilevel Monte Carlo \cite{giles,giles1} can help to reduce the computational cost of approximating expectations,  in comparison to single time discretization methods.  Using the ideas that have been developed in \cite{stoch_vol,jasra_cont,jasra_bpe_sde,med} we show how our particle MCMC algorithm can be extended to leverage upon the ideas of MLMC.

Our contributions are as follows:
\begin{itemize}
\item{We develop MCMC algorithms for PO MV SDES}
\item{The MCMC methods are extended to the multilevel context}
\item{We prove a bound on the mean square error (MSE) associated to our multilevel MCMC approach}
\item{We illustrate our method on two examples.}
\end{itemize}
In the context of the third bullet point the main result is as follows.  Let $\epsilon>0$ be given.  To achieve a MSE of $\mathcal{O}(\epsilon^2)$ our multilevel MCMC algorithm costs $\mathcal{O}(\epsilon^{-6})$.  If one uses a single level (i.e.~the MCMC algorithm in bullet point 1) then the cost to achieve the same MSE rises to $\mathcal{O}(\epsilon^{-7})$; i.e.~we reduce the cost by an order of magnitude.  Although these costs are large,  they are related to the fact that one has to approximate the biased posterior by MCMC (iteration cost),  then there is a cost of approximating the laws of the MV SDE and finally there is a time-discretization error.  It seems that controlling all these errors simultaneously is difficult without significant cost.   We remark that to further reduce the cost, 
as in \cite{nadir} one could adopt a multi-index approach,  which is a subject of future work.

This article is structured as follows.  In Section \ref{sec:problem} we detail the estimation problem that is to be considered.  In Section \ref{sec:method} we present the methodology to be used.
In Section \ref{sec:theory} our theoretical results are given.  Section \ref{sec:numerics} features our numerical results.  The proofs associated to our theoretical results can be found in Appendix \ref{app:appendix}.

\section{Problem Formulation}\label{sec:problem}

\subsection{Model}

We consider the stochastic differential equation (SDE) with $X_0=x_0\in\mathbb{R}^d$, $\theta\in\Theta\subseteq\mathbb{R}^{d_{\theta}}$ fixed:
\begin{equation}\label{eq:sde}
dX_t = a_{\theta}\left(X_t,\overline{\xi}_{\theta}(X_t,\mu_{t,\theta})\right)dt + \sigma\left(X_t\right)dW_t
\end{equation}
where 
$$
\overline{\xi}_{\theta}(X_t,\mu_{t,\theta})  =  \int_{\mathbb{R}^d}\xi_{\theta}(X_t,x)\mu_{t,\theta}(dx)
$$
$\{W_t\}_{t\geq 0}$ is a standard $d-$dimensional Brownian motion,  for each $\theta\in\Theta$,  $\xi_{\theta}:\mathbb{R}^{2d}\rightarrow\mathbb{R}$, $a_{\theta}:\mathbb{R}^d\times\mathbb{R}\rightarrow\mathbb{R}^d$, $\sigma:\mathbb{R}^d\rightarrow\mathbb{R}^{d\times d}$, $\mu_{t,\theta}$ is the law of $X_t$ and $\mu_{0,\theta}(dx)=\delta_{\{x_0\}}(dx)$ (Dirac measure on the set $\{x_0\}$ the starting point is taken as fixed). Note that conceputually,  there is no issue to allow $\sigma$ to depend upon $\theta$.  However,  as we do not consider this in our numerical examples,  we leave $\sigma$ independent of $\theta$ for ease of notation.

We make the following assumption throughout the paper; it is known \cite{ver} that this ensures the existence of a strong solution.  Set $\mathcal{P}(\mathbb{R}^d)$  the probability measures on the measurable space $(\mathbb{R}^d,\mathcal{B}(\mathbb{R}^d))$ with $\mathcal{B}(\mathbb{R}^d)$ the Borel sets on $\mathbb{R}^d$. We write $\mathcal{C}_b^k(\mathbb{R}^{d_1},\mathbb{R}^{d_2})$ as the collection of $k-$times continuously differentiable functions from $\mathbb{R}^{d_1}$ to $\mathbb{R}^{d_2}$ with bounded derivatives of order 1 up-to $k$.  Also,   $\mathcal{B}_b(\mathbb{R}^{d_1},\mathbb{R}^{d_2})$ as the collection of measurable functions from $\mathbb{R}^{d_1}$ to $\mathbb{R}^{d_2}$ which are bounded.

\begin{hypD}
\begin{enumerate}
\item{For each $(\theta,\mu)\in\Theta\times\mathcal{P}(\mathbb{R}^d)$, $(a_{\theta}\left(\cdot,\overline{\xi}_{\theta}(\cdot,\mu)\right)
,\xi_{\theta}(\cdot))
\in\mathcal{C}_b^2(\mathbb{R}^{d+1},\mathbb{R}^d)\cap\mathcal{B}_b(\mathbb{R}^{d+1},\mathbb{R}^d)
\times\mathcal{C}_b^2(\mathbb{R}^{2d},\mathbb{R})\cap\mathcal{B}_b(\mathbb{R}^{2d},\mathbb{R})$.}
\item{$\sigma(\cdot)\in\mathcal{C}_b^2(\mathbb{R}^{d},\mathbb{R}^{d\times d})\cap
\mathcal{B}_b(\mathbb{R}^{d},\mathbb{R}^{d\times d})$.}
\item{Set $\Sigma(x)=\sigma(x)\sigma(x)^{\top}$, then for any $x\in\mathbb{R}^d$, $\Sigma(x)$ is positive definite.}
\end{enumerate}
\end{hypD}

We denote by $P_{\mu_{t-1,\theta},t,\theta}(x_{t-1},dx_t)$ the conditional law of $X_t$ (as given in \eqref{eq:sde}) given $\mathscr{F}_{t-1}$ (the natual filtration of the process), for $t\geq 1$; that is, the transition kernel over unit time.  We consider a discrete time observation
process $Y_1,Y_2,\dots$,  $Y_T$,  $Y_t\in\mathsf{Y}$, $t\in\mathbb{N}$,  that are assumed, for notational convenience, to be observed at unit times. Conditional on the position $X_t$, $t\in\mathbb{N}$ of \eqref{eq:sde}, the random variable $Y_t$ is assumed to be independent of all other random variables, with a bounded and positive probability density $G_{\theta}(x_t,y_t)$.

Let $\nu$ be a Lebesgue probability density on the space $\Theta$.   Then our objective is to sample from the probability measure:
\begin{equation}\label{eq:post}
\pi\left(d(\theta,x_{1:T})|y_{1:T}\right) = \frac{\left\{\prod_{k=1}^T G_{\theta}(x_k,y_k)P_{\mu_{k-1,\theta},k,\theta}(x_{k-1},dx_k)\right\}\nu(\theta)d\theta}
{
\int_{\Theta\times\mathbb{R}^{Td}}
\left\{\prod_{k=1}^tTG_{\theta}(x_k,y_k)P_{\mu_{k-1,\theta},k,\theta}(x_{k-1},dx_k)
\right\}
\nu(\theta)d\theta
}
\end{equation}
where $d\theta$ is the Lebesgue measure on $(\Theta,\mathscr{B}(\Theta))$ and $\mathscr{B}(\Theta)$ is the Borel $\sigma-$field generated by $\Theta$.  It is assumed that the denominator in \eqref{eq:post} is finite.

\subsection{Time Discretization}

In practice,  one cannot work with the continuous-time formulation above,  so we shall consider a standard time-discretization and an associated optimization problem.  This in turn will be approximated by using numerical methods,  which will be the topic of Section \ref{sec:method}.

Let $l\in\mathbb{N}_0=\mathbb{N}\cup\{0\}$ be given and set $\Delta_l=2^{-l}$.  We consider the first order Euler-Maruyama time discretization of \eqref{eq:sde} for $k\in\{0,1\dots,\Delta_l^{-1}T-1\}$,  $\widetilde{x}_0=x_0$ (the starting position of \eqref{eq:sde} is a known point $x_0$)
\begin{equation}\label{eq:sde_disc}
\widetilde{X}_{(k+1)\Delta_l} = \widetilde{X}_{k\Delta_l} + a_{\theta}\left(\widetilde{X}_{k\Delta_l},\overline{\xi}_{\theta}(\widetilde{X}_{k\Delta_l},\mu_{k\Delta_l,\theta}^l)\right)dt + \sigma\left(\widetilde{X}_{k\Delta_l}\right)[W_{(k+1)\Delta_l}-W_{k\Delta_l}]
\end{equation}
where $\mu_{k\Delta_l,\theta}^l$ denotes the law of the time discretized process at time $k\Delta_l$. 
Recall $P_{\mu_{t-1,\theta},t,\theta}(x_{t-1},dx_t)$ ($t\in\{1,\dots,T\}$) from the previous section: we denote by
$P_{\mu_{t-1,\theta}^l,t,\theta}^l(x_{t-1},dx_t)$
the conditional law of $\widetilde{X}_t$ (as given in \eqref{eq:sde_disc}); the time discretized transition kernel over unit time. 

We then focus upon sampling from
\begin{equation}\label{eq:post_disc}
\pi^l\left(d(\theta,x_{1:T})|y_{1:T}\right) = \frac{\left\{\prod_{k=1}^T G_{\theta}(x_k,y_k)P_{\mu_{k-1,\theta}^l,k,\theta}^l(x_{k-1},dx_k)\right\}\nu(\theta)d\theta}
{
\int_{\Theta\times\mathbb{R}^{Td}}
\left\{\prod_{k=1}^T G_{\theta}(x_k,y_k)P_{\mu_{k-1,\theta}^l,k,\theta}^l(x_{k-1},dx_k)
\right\}
\nu(\theta)d\theta
}
\end{equation}
Even with using Euler-Maruyama time discretization,  we will need to appeal to further approximation,  which will be simulation-based.  This leads us to the next section.

\section{Methodology}\label{sec:method}

\subsection{Introduction}

We now describe the methodology that we will use to obtain, under further approximation,
samples from \eqref{eq:post_disc}.
The main complication that we need to address is the fact that numerical simulation of \eqref{eq:sde_disc} is not feasible without the approximating the laws $\mu_{k\Delta_l,\theta}^l$.   Once this can be achieved,  we will focus upon a novel extension of the work in \cite{jasra_bpe_sde} (see also \cite{mv_pe}).
We begin in Section \ref{sec:law_approx} by describing the method of \cite{basic_method} for approximating $\mu_{k\Delta_l,\theta}^l$.  Then using the afore-mentioned approach,  we develop an MCMC method to sample from an approximation of \eqref{eq:post_disc}; the details are in Section \ref{sec:smooth_approx}. 
Then we show how the MCMC method can be extended to be used in a multilevel context in 
Section \ref{sec:smooth_pair_approx}).  The final approach is summarized in Section \ref{sec:fa}.
Throughout the section the level $l$ of time discretization will be fixed unless otherwise stated.

\subsection{Approximating the Laws}\label{sec:law_approx}

Throughout this section $\theta\in\Theta$ is fixed.
We now give the method of \cite{basic_method} for allowing one to approximate the law
$\mu_{t,\theta}^l$ for each $t\in\{1,\dots,T\}$.  We will need the notation that
$\mathcal{N}_d(c,\Sigma)$ denotes the $d-$dimensional Gaussian distribution with mean $c$ and covariance matrix $\Sigma$. $I_d$ is the $d\times d$
identity matrix and $\stackrel{\textrm{ind}}{\sim}$ denotes independently distributed as.
The approach is given in Algorithm \ref{alg:basic_method} from \cite{basic_method}.   
Note that in Step 2.~when $k=1$ the points $X_{t-1}^1,\dots,X_{t-1}^N$ are available via the empirical
measure that has to be specified in Step 1..
We remark that
Algorithm \ref{alg:basic_method} can be used to approximate expectations w.r.t.~$\mu_{t,\theta}^l$
and indeed on the discrete grid in-between time $t-1$ and $t$.

\begin{algorithm}[h]
\begin{enumerate}
\item{Input $l\in\mathbb{N}_0$ the level of discretization, $N\in\mathbb{N}$ the number of particles, $\theta\in\Theta$,  $t\in\{1,\dots,T\}$. If $t=1$ set $\mu_{0,\theta}^{l,N}(dx)=\delta_{\{x_0\}}(dx)$ otherwise input an empirical measure $\mu_{t-1,\theta}^{l,N}(dx)=\tfrac{1}{N}\sum_{i=1}^N\delta_{\{X_{t-1}^i\}}(dx)$. Set $k=1$.}
\item{For $i\in\{1,\dots,N\}$ generate:
\begin{align*}
X_{t-1+k\Delta_l}^i & =  X_{t-1+(k-1)\Delta_l}^i + a_{\theta}\left(X_{t-1+(k-1)\Delta_l}^i,\overline{\xi}_{\theta}(X_{t-1+(k-1)\Delta_l}^i,\mu_{t-1+(k-1)\Delta_l,\theta}^{l,N})\right) + \\ & \sigma\left(X_{t-1+(k-1)\Delta_l}^i\right)\left[W_{t-1+k\Delta_l}^i - W_{t-1+(k-1)\Delta_l}^i\right]
\end{align*}
where
\begin{eqnarray*}
\overline{\xi}_{\theta}(X_{t-1+(k-1)\Delta_l}^i,\mu_{t-1+(k-1)\Delta_l,\theta}^{l,N}) & = & \frac{1}{N}\sum_{j=1}^N \xi_{\theta}(X_{t-1+(k-1)\Delta_l}^i,X_{t-1+(k-1)\Delta_l}^j)\\
\mu_{t-1+(k-1)\Delta_l,\theta}^{l,N}(dx) & = & \frac{1}{N}\sum_{j=1}^N\delta_{\{X_{t-1+(k-1)\Delta_l}^j\}}(dx) \\
\left[W_{t-1+k\Delta_l}^i - W_{t-1+(k-1)\Delta_l}^i\right] & \stackrel{\textrm{ind}}{\sim} & \mathcal{N}_{d}(0,\Delta_l I_d).
\end{eqnarray*}
Set $k=k+1$, if $k=\Delta_l^{-1}+1$ go to step 3.~otherwise go to the start of step 2..}
\item{Output all the required laws $\mu_{t-1+\Delta_l,\theta}^N,\dots,\mu_{t,\theta}^N$.}
\end{enumerate}
\caption{Approximating the Laws when starting with a particle approximation at time $t-1$, $t\in\{1,\dots,T\}$.}
\label{alg:basic_method}
\end{algorithm}

\subsection{Markov Chain Monte Carlo Approach}\label{sec:smooth_approx}

\subsubsection{Modified Posterior}\label{sec:mod_targ}

We now describe,  by using the approach in Algorithm \ref{alg:basic_method},  an approximation
of the time-discretized posterior in \eqref{eq:post_disc}.  This is the probability measure that we will seek to sample from.
We write
$Q_{\mu,\theta}^l(x,dy)$ as the Gaussian Markov kernel on $(\mathbb{R}^d,\mathscr{B}(\mathbb{R}^d))$ associated to a $\Delta_l$ time step
of \eqref{eq:sde_disc},  with input measure $\mu\in\mathcal{P}(\mathbb{R}^d)$. That is,  for any $t\in\{1,\dots,T\}$, we have
$$
P_{\mu_{t-1,\theta}^l,t,\theta}^l(x_{t-1},dx_t) = \int_{\mathbb{R}^{d(\Delta_l^{-1}-1)}}
\prod_{k=1}^{\Delta_{l}^{-1}} Q_{\mu_{t-1+(k-1)\Delta_l,\theta}^l,\theta}^l(x_{t-1+(k-1)\Delta_l},dx_{t-1+k\Delta_l}).
$$
 Set $u_t=(x_{t-1+\Delta_l},\dots,x_t)$,   and write
$$
\overline{P}_{\mu_{t-1,\theta}^l,t,\theta}^l(x_{t-1},du_t) := \prod_{k=1}^{\Delta_{l}^{-1}} Q_{\mu_{t-1+(k-1)\Delta_l,\theta}^l,\theta}^l(x_{t-1+(k-1)\Delta_l},dx_{t-1+k\Delta_l})
$$
then the following probability measure on $(\Theta\times\mathsf{E}_l=\mathbb{R}^{d\Delta_l^{-1}},\mathscr{B}(\Theta\times\mathsf{E}_l))$
$$
\overline{\pi}^l\left(d(\theta,u_1,\dots,u_T)\right) := \frac{\left\{\prod_{k=1}^T G_{\theta}(x_k,y_k)
\overline{P}_{\mu_{k-1,\theta}^l,k,\theta}^l(x_{k-1},du_k)\right\}\nu(\theta)d\theta
}{\int_{\Theta\times\mathsf{E}_l^T}
\left\{\prod_{k=1}^T G_{\theta}(x_k,y_k)
\overline{P}_{\mu_{k-1,\theta}^l,k,\theta}^l(x_{k-1},du_k)\right\}\nu(\theta)d\theta
}
$$
admits  \eqref{eq:post_disc} as a marginal.
Ideally,  our objective would now to be to sample from $\overline{\pi}^l$.  The issue here is that of course we do not have access to the laws.

Consider sequentially sampling Algorithm \ref{alg:basic_method} from time
$1$ to time $T$ using the approximated laws from the previous time step,  to initialize the next time step; throughout $N$ is fixed and so is $\theta$.
At time $t$, write all the simulated random variables $(X_{t-1+\Delta_l}^1,\dots,X_{t-1+\Delta_l}^N),\dots,
(X_{t}^1,\dots,X_{t}^N)$ from Algorithm \ref{alg:basic_method} as $\overline{u}_t=u_t^{1:N}\in\mathsf{E}_l^N$.
Finally write the joint probability of $(\overline{u}_1,\dots,\overline{u}_T)$ as $\overline{\mathbb{P}}_{\theta}^N$
and associated expectation $\overline{\mathbb{E}}_{\theta}^N$. 
We will write for $t\in\{1,\dots,T\}$
$$
\overline{P}_{\mu_{t-1,\theta}^{l,N},t,\theta}^l(x_{t-1},du_t) =  \prod_{k=1}^{\Delta_{l}^{-1}} Q_{\mu_{t-1+(k-1)\Delta_l,\theta}^{l,N},\theta}^l(x_{t-1+(k-1)\Delta_l},dx_{t-1+k\Delta_l}).
$$
We will now consider an MCMC method to sample from
$$
\overline{\pi}^{l,N}\left(d(\theta,u_1,\dots,u_T)\right) = \frac{
\left\{\prod_{k=1}^T G_{\theta}(x_k,y_k)\right\}
\overline{\mathbb{E}}_{\theta}^N\left[\prod_{k=1}^T \overline{P}_{\mu_{k-1,\theta}^{l,N},t,\theta}^l(x_{k-1},du_k)
\right]\nu(\theta)d\theta
}{
\int_{\Theta\times \mathsf{E}_l^T}
\left\{\prod_{k=1}^T G_{\theta}(x_k,y_k)\right\}
\overline{\mathbb{E}}_{\theta}^N\left[
\prod_{k=1}^T\overline{P}_{\mu_{k-1,\theta}^{l,N},t,\theta}^l(x_{k-1},du_k)
\right]\nu(\theta)d\theta
}.
$$

\subsubsection{MCMC Method}

Our MCMC method is given in Algorithm \ref{alg:pmmh_0} which needs as an input the particle filter in 
Algorithm \ref{alg:cond_pf_0}.  
In Algorithm \ref{alg:pmmh_0} step 4.~$\mathcal{U}_{[0,1]}$ is the uniform distribution on $[0,1]$.
The approach is simply the particle marginal Metropolis-Hastings method
that was originally in \cite{andrieu}, except that it is adapted to the case here.  The particle filter that we
use is simply that of \cite{po_mv}. 
Using the approach that was developed in \cite{andrieu},  one can show that under minimal conditions that the sampling from $K_{l}$  (the kernel in Algorithm \ref{alg:pmmh_0}) provides a means to obtain samples from $\overline{\pi}^{l,N}_{\theta}$.  
We remark the cost of applying
$K_{l,\theta}$ is $\mathcal{O}(\Delta_l^{-1}N^2T)$ in step 2.~of Algorithm  \ref{alg:cond_pf_0} and then
the rest of the algorithm has a cost of $\mathcal{O}(\Delta_l^{-1}NMT)$.

On identifying an initial state of $\theta,u_{1:T}$,  for instance sampling $\theta$ from the prior and then running Algorithm \ref{alg:cond_pf_0},  one can iteratively apply Algorithm \ref{alg:pmmh_0}.  Denote the sample
values as $\theta(j),u_{1:T}(j)$ at iteration $j$,  where $j\in\{0,1,\dots,I\}$, then one can approximate, for $\varphi:\Theta\times\mathbb{R}^{Td}\rightarrow\mathbb{R}$,  and assuming it is well-defined:
$$
\int_{\Theta\times\mathsf{E}_l^T}\varphi(\theta,x_{1:T}) \overline{\pi}^{l,N}\left(d(\theta,u_1,\dots,u_T)\right)
$$
by using
\begin{equation}\label{eq:mcmc_est}
\overline{\pi}^{l,N,I}(\varphi) := \frac{1}{I+1}\sum_{j=0}^I \varphi(\theta(j),x_{1:T}(j)).
\end{equation}

\begin{algorithm}[h]
\begin{enumerate}
\item{Input $l\in\mathbb{N}_0$ the level of discretization, $N\in\mathbb{N}$ the number of particles for the approach of Algorithm \ref{alg:basic_method}, $M$ the number of particles for the main method and $\theta\in\Theta$. }
\item{Sample $(\overline{u}_1,\dots,\overline{u}_T)$ from $\overline{\mathbb{P}}_{\theta}^N$ via a sequential
application of Algorithm \ref{alg:basic_method}, with $N$ particles.}
\item{Initialize: Sample $U_1^i$ independently from $
\prod_{k=1}^{\Delta_{l}^{-1}} Q_{\mu_{(k-1)\Delta_l,\theta}^{l,N},\theta}^l(x_{(k-1)\Delta_l},dx_{k\Delta_l})$, 
$A_{0}^i=i$ for $i\in\{1,\dots,M\}$. Set $p_{\theta}^{l,M,N}(y_{-1})=1$, $k=1$.}
\item{Resampling: Construct the probability mass function on $\{1,\dots,M\}$:
$$
r_1^i = \frac{G_{\theta}(x_k^i,y_k)}{\sum_{j=1}^MG_{\theta}(x_k^j,y_k)}.
$$
For $i\in\{1,\dots,M\}$ sample $A_k^i$ from $r_1^i$. Set 
$p_{\theta}^{l,M,N}(y_{1:k}) = p_{\theta}^{l,M,N}(y_{1:k-1})\tfrac{1}{M} \sum_{j=1}^MG_{\theta}(x_k^j,y_k)$ and
$k=k+1$.
}
\item{Sampling: for $i\in\{1,\dots,M\}$ sample $U_{k}^i|u_{k-1}^{A_{k-1}^i}$ using the kernel 
$$ 
\prod_{k=1}^{\Delta_{l}^{-1}} Q_{\mu_{t-1+(k-1)\Delta_l,\theta}^{l,N},\theta}^l(x_{t-1+(k-1)\Delta_l},dx_{t-1+k\Delta_l})
$$
where $x_{t-1}^i=x_{t-1}^{A_{k-1}^i}$.  
For $i\in\{1,\dots,M\}$, $(u_{1}^i,\dots,u_k^i)=(u_1^{A_{k-1}^i},\dots,u_{k-1}^{A_{k-1}^i},u_{k}^i)$. 
If $k=T$ go to 6.  otherwise return to the start of 4..}
\item{Construct the probability mass function on $\{1,\dots,M\}$:
$$
r_1^i = \frac{G_{\theta}(x_T^i,y_T)}{\sum_{j=1}^M G_{\theta}(x_T^j,y_T)}.
$$
Set $p_{\theta}^{l,M,N}(y_{1:k}) = p_{\theta}^{l,M,N}(y_{1:k-1})\tfrac{1}{M} \sum_{j=1}^MG_{\theta}(x_k^j,y_k)$.
Sample $i\in\{1,\dots,M\}$ using $r_1$ and return $(u_1^i,\dots,u_T^i)$,  $p_{\theta}^{l,M,N}(y_{1:T})$.}
\end{enumerate}
\caption{Particle Filter at level $l\in\mathbb{N}_0$.}
\label{alg:cond_pf_0}
\end{algorithm}

\begin{algorithm}[h]
\begin{enumerate}
\item{Input $l\in\mathbb{N}_0$ the level of discretization, $N\in\mathbb{N}$ the number of particles for the approach of Algorithm \ref{alg:basic_method}, $M$ the number of particles for the main method and $\theta\in\Theta$,  $p_{\theta}^{l,M,N}(y_{1:T})$, $(u_1,\dots,u_T)\in\mathsf{E}_l^T$. }
\item{Propose $\theta'|\theta$ from a proposal kernel $r_2(\theta,\theta')d\theta'$.}
\item{Run Algorithm \ref{alg:cond_pf_0} with the given $\theta'$ and other input parameters.}
\item{Generate $Z\sim\mathcal{U}_{[0,1]}$
and return $(\theta',p_{\theta'}^{l,M,N}(y_{1:T}))$ and $(u_1',\dots,u_T')$ if
$$
Z<\min\left\{
1, \frac{p_{\theta'}^{l,M,N}(y_{1:T})\nu(\theta')r_2(\theta',\theta)}
{
p_{\theta}^{l,M,N}(y_{1:T})\nu(\theta)
r_2(\theta,\theta')}
\right\}
$$
otherwise return $(\theta,p_{\theta}^{l,M,N}(y_{1:T}))$ and $(u_1,\dots,u_T)$.
}
\end{enumerate}
\caption{Particle Marginal Metropolis Kernel at level $l\in\mathbb{N}_0$.}
\label{alg:pmmh_0}
\end{algorithm}

\subsection{Bi-Level MCMC}\label{sec:smooth_pair_approx}

\subsubsection{Set-Up}

Let $l\in\mathbb{N}$ be fixed.
Let $\varphi:\Theta\times\mathbb{R}^{Td}\rightarrow\mathbb{R}$ now consider the approximation of
the following formula,  assuming it is well-defined:
\begin{equation}\label{eq:diff}
\int_{\Theta\times\mathsf{E}_l^T}\varphi(\theta,x_{1:T})
\overline{\pi}^{l,N_l}\left(d(\theta,u_1,\dots,u_T)\right)
-
\int_{\Theta\times\mathsf{E}_{l-1}^T}\varphi(\theta,x_{1:T})
\overline{\pi}^{l-1,N_{l}}\left(d(\theta,u_1,\dots,u_T)\right)
\end{equation}
with $N_{l}\in\mathbb{N}$.  Differences of this type are critical to using multilevel Monte Carlo (MLMC) approximation,  which is our objective.  Typically, to use MLMC,  one seeks to sample
from a coupling of $\overline{\pi}^{l,N_l}$ and $\overline{\pi}^{l-1,N_{l}}$ and specifically a MCMC kernel which could achieve such a simulation.  In our context,  we will sample from a joint probability measure on the space
$(\Theta\times\mathsf{E}_{l}^T\times\mathsf{E}_{l-1}^T,\mathscr{B}(\Theta\times\mathsf{E}_{l}^T\times\mathsf{E}_{l-1}^T))$ which by an appropriate change of measure formula will allow us to approximate \eqref{eq:diff}
using Markov chain Monte Carlo.  This is the approach that was adopted in \cite{jasra_bpe_sde},  except it is modified to the context here.

In order to proceed,  for a given $t\in\{1,\dots,T\}$,  we will need a method to approximate the laws
$\mu_{t-1+\Delta_l,\theta}^l,\dots,\mu_{t,\theta}^l$ and simultaneously $\mu_{t-1+\Delta_{l-1},\theta}^{l-1},\dots,\mu_{t,\theta}^{l-1}$;
to avoid confusion we will denote the latter as $\widetilde{\mu}_{t-1+\Delta_l,\theta}^{l-1},\dots,\widetilde{\mu}_{t,\theta}^{l-1}$.  A method for approximating these measures is given in 
Algorithm \ref{alg:basic_method_coup},  which was originally developed in \cite{po_mv}.  The approach given in 
Algorithm \ref{alg:basic_method_coup} is simply a coupled version of Algorithm \ref{alg:basic_method}.  The reason for using a dependent coupling is to help ensure that the estimates of \eqref{eq:diff} to be developed will have a variance that is suitably falling as $l$ grows;  this is the core principal of using MLMC.

We now seek to use an MCMC method to approximate \eqref{eq:diff},  which will follow the ideas that were originally proposed in \cite{jasra_bpe_sde}. 
Consider sequentially sampling Algorithm \ref{alg:basic_method_coup} from time
$1$ to time $T$ using the approximated laws from the previous time step,  to initialize the next time step; throughout $N_l$ is fixed and so is $\theta$.  We will write the joint probability of 
$(\overline{u}_1^l,\dots,\overline{u}_T^l)$ and
$(\widetilde{\overline{u}}_1^{l-1},\dots,\widetilde{\overline{u}}_T^{l-1})$
as $\check{\mathbb{P}}_{\theta}^{N_{l}}$
and associated expectation $\check{\mathbb{E}}_{\theta}^{N_{l}}$.   Here we have
$\overline{u}_t=u_t^{1:N_l}\in\mathsf{E}_l^{N_l}$ and
$\widetilde{\overline{u}}_t=\widetilde{u}_t^{1:N_{l}}\in\mathsf{E}_{l-1}^{N_{l}}$.  Now for $(x,x',\theta)\in\mathbb{R}^{2d}\times\Theta$ and $k\in\{1,\dots,T\}$ set
\begin{equation}\label{eq:hk_ch}
H_{k,\theta}(x,x') = \tfrac{1}{2}\left\{G_{\theta}(x',y_k)+G_{\theta}(x,y_k)\right\}.
\end{equation}
The next idea we need is the simulation of a synchronous coupling of two first order Euler-Maruyama time discretizations of 
$$
 \overline{P}_{\mu_{t-1,\theta}^{l,N_l},t,\theta}^l(x_{t-1},du_t)
\quad \textrm{and} \quad
 \overline{P}_{\mu_{t-1,\theta}^{l-1,N_{l}},t,\theta}^{l-1}(x_{t-1},du_t)
$$
where the input probability measures are simulated via Algorithm \ref{alg:basic_method_coup} independently
of all other random variables represented in the displayed equation. This has,  in effect, been described in 
Algorithm \ref{alg:basic_method_coup},  except of course we need only one particle ($N_l=N_{l-1}=1$ in Algorithm \ref{alg:basic_method_coup}) and the empirical measures used to approximate the law of the MV SDE need not be updated,  as they have been approximated and plugged in already.  Therefore we will write such a simulation of $(u_t^l,u_{t}^{l-1})\in\mathsf{E}_l\times\mathsf{E}_{l-1}$
conditonal upon $(u_{t-1}^l,u_{t-1}^{l-1})$ (where $(u_{0}^l,u_{0}^{l-1})=(x_0,x_0))$ as the kernel $\check{P}^{l}_{\theta}$. 
Then we define the following probability measure on $\left(\Theta\times(\mathsf{E}_l\times\mathsf{E}_{l-1})^T,\mathscr{B}(\Theta\times(\mathsf{E}_l\times\mathsf{E}_{l-1})^T)\right)$
$$
\check{\pi}^{l,N_{l}}\left(d(\theta,u_{1:T}^l,\widetilde{u}_{1:T}^{l-1})\right) = 
$$
\begin{equation}\label{eq:main_tar}
\frac{\left\{\prod_{k=1}^T H_{k,\theta}(x_k^l,\widetilde{x}_k^{l-1})\right\}
\check{\mathbb{E}}_{\theta}^{N_{l}}\left[
\prod_{k=1}^T \check{P}_{\theta}^l\left((x_{k-1}^l,\widetilde{x}_{k-1}^{l-1}),d(u_k^l,\widetilde{u}_{k}^{l-1})\right)
\right]\nu(\theta)d\theta
}{
\int_{\Theta\times(\mathsf{E}_l\times\mathsf{E}_{l-1})^T}
\left\{\prod_{k=1}^T H_{k,\theta}(x_k^l,\widetilde{x}_k^{l-1})\right\}
\check{\mathbb{E}}_{\theta}^{N_{l}}\left[
\prod_{k=1}^T \check{P}_{\theta}^l\left((x_{k-1}^l,\widetilde{x}_{k-1}^{l-1}),d(u_k^l,\widetilde{u}_{k}^{l-1})\right)
\right]\nu(\theta)d\theta
}.
\end{equation}

Given the infomation above,  we have that \eqref{eq:diff} is equal to
$$
\frac{\int_{\Theta\times(\mathsf{E}_l\times\mathsf{E}_{l-1})^T}\varphi(\theta,x_{1:T}^l)
\left\{\prod_{k=1}^T \check{H}_{k,\theta}(x_k^l,\widetilde{x}_{k}^{l-1})\right\}
\check{\pi}^{l,N_{l}}\left(d(\theta,u_{1:T}^l,\widetilde{u}_{1:T}^{l-1})\right)
}
{\int_{\Theta\times(\mathsf{E}_l\times\mathsf{E}_{l-1})^T}
\left\{\prod_{k=1}^T \check{H}_{k,\theta}(x_k^l,\widetilde{x}_{k}^{l-1})\right\}
\check{\pi}^{l,N_{l}}\left(d(\theta,u_{1:T}^l,\widetilde{u}_{1:T}^{l-1})\right)} - 
$$
\begin{equation}\label{eq:main_eq}
\frac{\int_{\Theta\times(\mathsf{E}_l\times\mathsf{E}_{l-1})^T}\varphi(\theta,\widetilde{x}_{1:T}^{l-1})
\left\{\prod_{k=1}^T \check{H}_{k,\theta}(\widetilde{x}_{k}^{l-1},x_k^l)\right\}
\check{\pi}^{l,N_{l}}\left(d(\theta,u_{1:T}^l,\widetilde{u}_{1:T}^{l-1})\right)
}
{\int_{\Theta\times(\mathsf{E}_l\times\mathsf{E}_{l-1})^T}
\left\{\prod_{k=1}^T \check{H}_{k,\theta}(\widetilde{x}_{k}^{l-1},x_k^l)\right\}
\check{\pi}^{l,N_{l}}\left(d(\theta,u_{1:T}^l,\widetilde{u}_{1:T}^{l-1})\right)}
\end{equation}
where,  for $(x,x',\theta)\in\mathbb{R}^{2d}\times\Theta$ and $k\in\{1,\dots,T\}$
$$
\check{H}_{k,\theta}(x,x') = \frac{G_{\theta}(x,y_k)}{H_{k,\theta}(x,x')}.
$$
Given \eqref{eq:main_eq} the objective is then to obtain an MCMC algorithm to sample from 
\eqref{eq:main_tar}.  Before we describe the MCMC method,  we pause to make a couple of remarks.
The basic motivation is to try and sample from a coupling of $\overline{\pi}^{l,N_l}$ and $\overline{\pi}^{l-1,N_{l}}$ using MCMC.  However, to the best of our knowledge,  with the exception of sampling the independent coupling - which is not useful,  there is no MCMC method to achieve this goal.  The idea in 
\cite{jasra_bpe_sde} was then to sample from a joint probability that ensures that there is some coupling between the $(u_{1:T}^l,\widetilde{u}_{1:T}^{l-1})$ in manner which allows a change of measure,  as is used in 
\eqref{eq:main_eq} which can ensure that the weight functions:
$$
\left\{\prod_{k=1}^T \check{H}_{k,\theta}(\widetilde{x}_{k}^{l-1},x_k^l)\right\}
$$
have a stable variance (w.r.t.~$T$ and $l$) and facilitate the efficient application of the MLMC paradigm.
The choice \eqref{eq:hk_ch} is the same as used in \cite{chada_ub} and some reasoning can be found there.

\begin{algorithm}[h]
\begin{enumerate}
\item{Input $l\in\mathbb{N}$ the level of discretization, $N_l\in\mathbb{N}$ the number of particles, $\theta\in\Theta$
,  $t\in\{1,\dots,T\}$. If $t=1$ set $\mu_{0,\theta}^{l,N_{l}}(dx)=\widetilde{\mu}_{0,\theta}^{l-1,N_{l}}(dx)=\delta_{\{x_0\}}(dx)$ otherwise input a pair of empirical measures $\mu_{t-1,\theta}^{l,N_l}(dx)=\tfrac{1}{N_l}\sum_{i=1}^{N_l}\delta_{\{X_{t-1}^{l,i}\}}(dx)$, 
$\widetilde{\mu}_{t-1,\theta}^{l-1,N_{l}}(dx)=\tfrac{1}{N_{l}}\sum_{i=1}^{N_{l}}\delta_{\{\widetilde{X}_{t-1}^{l-1,i}\}}(dx)$.  Set $k=1$.}
\item{For $i\in\{1,\dots,N_l\}$ generate:
\begin{align*}
X_{t-1+k\Delta_l}^{l,i} &  =  X_{t-1+(k-1)\Delta_l}^{l,i} + a_{\theta}\left(X_{t-1+(k-1)\Delta_l}^{l,i},\overline{\xi}_{\theta}(X_{t-1+(k-1)\Delta_l}^{l,i},\mu_{t-1+(k-1)\Delta_l,\theta}^{l,N_l})\right) + \\ & \sigma\left(X_{t-1+(k-1)\Delta_l}^{l,i}\right)\left[W_{t-1+k\Delta_l}^i - W_{t-1+(k-1)\Delta_l}^i\right]
\end{align*}
where
\begin{eqnarray*}
\overline{\xi}_{\theta}(X_{t-1+(k-1)\Delta_l}^i,\mu_{t-1+(k-1)\Delta_l}^{l,N_l}) & = & \frac{1}{N_l}\sum_{j=1}^{N_l} \xi_{\theta}(X_{t-1+(k-1)\Delta_l}^{l,i},X_{t-1+(k-1)\Delta_l}^{l,j})\\
\mu_{t-1+(k-1)\Delta_l,\theta}^{l,N_l}(dx) & = & \frac{1}{N_l}\sum_{j=1}^{N_l}\delta_{\{X_{t-1+(k-1)\Delta_l}^{l,j}\}}(dx).
\end{eqnarray*}
Set $k=k+1$, if $k=\Delta_l^{-1}+1$ go to step 3.~otherwise go to the start of step 2..}
\item{For $i\in\{1,\dots,N_{l}\}$ compute:
\begin{align*}
\widetilde{X}_{t-1+k\Delta_{l-1}}^{l-1,i} & =  \widetilde{X}_{t-1+(k-1)\Delta_{l-1}}^{l-1,i} + a_{\theta}\left(\widetilde{X}_{t-1+(k-1)\Delta_{l-1}}^{l-1,i},\overline{\xi}_{\theta}(\widetilde{X}_{t-1+(k-1)\Delta_{l-1}}^{l-1,i},\widetilde{\mu}_{t-1+(k-1)\Delta_{l-1},\theta'}^{l-1,N_{l}})\right) + \\ & \sigma\left(\widetilde{X}_{t-1+(k-1)\Delta_{l-1}}^{l-1,i}\right)\left[W_{t-1+k\Delta_{l-1}}^i - W_{t-1+(k-1)\Delta_{l-1}}^i\right]
\end{align*}
where
\begin{eqnarray*}
\overline{\xi}_{\theta}(\widetilde{X}_{t-1+(k-1)\Delta_{l-1}}^{l-1,i},\widetilde{\mu}_{t-1+(k-1)\Delta_{l-1},\theta'}^{l-1,N_l}) & = & \frac{1}{N_{l}}\sum_{j=1}^{N_{l}} \xi_{\theta}(\widetilde{X}_{t-1+(k-1)\Delta_{l-1}}^{l-1,i},\widetilde{X}_{t-1+(k-1)\Delta_{l-1}}^{l-1,j})\\
\widetilde{\mu}_{t-1+(k-1)\Delta_{l-1},\theta'}^{l-1,N_{l}}(dx) & = & \frac{1}{N_{l}}\sum_{j=1}^{N_{l}}\delta_{\{\widetilde{X}_{t-1+(k-1)\Delta_l}^{l-1,j}\}}(dx)
\end{eqnarray*}
and the increments of the Brownian motion $\left[W_{t-1+k\Delta_{l-1}}^i - W_{t-1+(k-1)\Delta_{l-1}}^i\right]$ were generated in step 2..
Set $k=k+1$, if $k=\Delta_{l-1}^{-1}+1$ go to step 4.~otherwise go to the start of step 3..}
\item{Output all the required laws $\mu_{t-1+\Delta_l,\theta}^{l,N_l},\dots,\mu_{t,\theta}^{l,N_l}$,  $\widetilde{\mu}_{t-1+\Delta_l,\theta}^{l-1,N_{l}},\dots,\widetilde{\mu}_{t,\theta}^{l-1,N_{l}}$.}
\end{enumerate}
\caption{Approximating the Consecutive Laws when starting with a particle approximation at time $t-1$, $t\in\{1,\dots,T\}$.}
\label{alg:basic_method_coup}
\end{algorithm}

\subsubsection{Algorithm}

Our MCMC method for sampling from \eqref{eq:main_tar} is now presented.  The Markov kernel is given in Algorithm \ref{alg:pmmh_l} (denoted as $\check{K}_{l}$) which requires the delta particle filter (as termed in \cite{chada_ub}) in Algorithm \ref{alg:delta_pf} to be run.  We note that the cost of applying the kernel in Algorithm \ref{alg:pmmh_l} is $\mathcal{O}(\Delta_l^{-1}T\{MN_l + N_l^2\})$.

We can initialize our MCMC algorithm by sampling $\theta$ from the prior and then running Algorithm \ref{alg:delta_pf} and then iteratively applying Algorithm \ref{alg:pmmh_l}.  Denote the sample
values as $\theta^l(j),u_{1:T}(j)^l,\widetilde{u}_{1:T}^{l-1}(j)$ at iteration $j$,  where $j\in\{0,1,\dots,I\}$, then one can approximate \eqref{eq:main_eq} using the following:
$$
\overline{\pi}^{l,N_l,I}(\varphi) - \overline{\pi}^{l-1,N_{l},I}(\varphi) := 
$$
\begin{equation}\label{eq:mcmc_bl_est}
\frac{
\frac{1}{I+1}\sum_{j=0}^I\varphi(\theta^l(j),x_{1:T}^{l}(j))
\left\{\prod_{k=1}^T \check{H}_{k,\theta^l(j)}(x_k^l(j),\widetilde{x}_{k}^{l-1}(j))\right\}
}{\frac{1}{I+1}\sum_{j=0}^I
\left\{\prod_{k=1}^T \check{H}_{k,\theta^l(j)}(x_k^l(j),\widetilde{x}_{k}^{l-1}(j))\right\}} -
\frac{
\frac{1}{I+1}\sum_{j=0}^I\varphi(\theta^l(j),\widetilde{x}_{1:T}^{l-1}(j))
\left\{\prod_{k=1}^T \check{H}_{k,\theta^l(j)}(\widetilde{x}_{k}^{l-1}(j),x_k^l(j))\right\}
}{\frac{1}{I+1}\sum_{j=0}^I
\left\{\prod_{k=1}^T \check{H}_{k,\theta^l(j)}(\widetilde{x}_{k}^{l-1}(j),x_k^l(j))\right\}}.
\end{equation}

\begin{algorithm}[h]
\begin{enumerate}
\item{Input $l\in\mathbb{N}$ the level of discretization, $N_l\in\mathbb{N}$ the number of particles for the approach of Algorithm \ref{alg:basic_method_coup}, $M$ the number of particles for the main method and $\theta\in\Theta$. }
\item{Sample $(\overline{U}_{1:T}^l,\widetilde{\overline{U}}^{l-1}_{1:T})$ from 
$\check{\mathbb{P}}_{\theta}^{N_{l-1:l}}$ via a sequential
application of Algorithm \ref{alg:basic_method_coup}, with $(N_l,N_{l-1})$ particles.}
\item{Initialize: Sample $(U_1^{i,l},\widetilde{U}_1^{i,l-1})$ independently from $\check{P}_{\theta}^{l}\left((x_0,x_0), d(u_1^l,\widetilde{u}_1^{l-1})\right)$,  where the input measures have been generated in Step 2., 
$A_{0}^i=i$ for $i\in\{1,\dots,M\}$. Set $p_{\theta}^{l,M,N_{l}}(y_{-1})=1$, $k=1$.}
\item{Resampling: Construct the probability mass function on $\{1,\dots,M\}$:
$$
r_1^i = \frac{H_{k,\theta}(x_k^{i,l},\widetilde{x}_k^{i,l-1})}{\sum_{j=1}^MH_{k,\theta}(x_k^{j,l},\widetilde{x}_k^{j,l-1})}.
$$
For $i\in\{1,\dots,M\}$ sample $A_k^i$ from $r_1^i$. Set 
$p_{\theta}^{l,M,N_{l}}(y_{1:k}) = p_{\theta}^{l,M,N_{l}}(y_{1:k-1})\tfrac{1}{M} \sum_{j=1}^M
H_{k,\theta}(x_k^{j,l},\widetilde{x}_k^{j,l-1})
$ and
$k=k+1$.
}
\item{Sampling: for $i\in\{1,\dots,M\}$ sample $U_{k}^{i,l},\widetilde{U}_k^{i,l-1}|u_{k-1}^{A_{k-1}^i,l},\widetilde{u}_{k-1}^{A_{k-1}^i,,l-1}$ using the kernel 
$$\check{P}_{\theta}^{l}\left((x_{k-1}^{A_{k-1}^i,l},\widetilde{x}_{k-1}^{A_{k-1}^i,l-1}), d(u_k^l,\widetilde{u}_k^{l-1})\right)$$
where the input measures have been generated in Step 2..  
For $i\in\{1,\dots,M\}$, $(u_{1}^{i,l},\dots,u_k^{i,l})=(u_1^{A_{k-1}^i,l},\dots,u_{k-1}^{A_{k-1}^i,l},u_{k}^{i,l})$
and 
$(\widetilde{u}_{1}^{i,l-1},\dots,\widetilde{u}_k^{i,l-1})=(\widetilde{u}_1^{A_{k-1}^i,l-1},\dots,\widetilde{u}_{k-1}^{A_{k-1}^i,l-1},\widetilde{u}_{k}^{i,l-1})$.
If $k=T$ go to 6.  otherwise return to the start of 4..}
\item{Construct the probability mass function on $\{1,\dots,M\}$:
$$
r_1^i = \frac{H_{T,\theta}(x_T^{i,l},\widetilde{x}_T^{i,l-1})}{\sum_{j=1}^M H_{T,\theta}(x_T^{j,l},\widetilde{x}_T^{j,l-1})}.
$$
Set $p_{\theta}^{l,M,N_{l}}(y_{1:k}) = p_{\theta}^{l,M,N_{l}}(y_{1:k-1})\tfrac{1}{M} \sum_{j=1}^M
H_{T,\theta}(x_T^{j,l},\widetilde{x}_T^{j,l-1})$.
Sample $i\in\{1,\dots,M\}$ using $r_1$ and return $(u_1^{i,l},\dots,u_T^{i,l-1})$,  
$(\widetilde{u}_1^{i,l},\dots,\widetilde{u}_T^{i,l-1})$,  
$p_{\theta}^{l,M,N_{l}}(y_{1:T})$.}
\end{enumerate}
\caption{Delta Particle Filter at level $l\in\mathbb{N}$.}
\label{alg:delta_pf}
\end{algorithm}

\begin{algorithm}[h]
\begin{enumerate}
\item{Input $l\in\mathbb{N}$ the level of discretization, $N_l\in\mathbb{N}$ the number of particles for the approach of Algorithm \ref{alg:basic_method_coup}, $M$ the number of particles for  
Algorithm \ref{alg:delta_pf}
and $\theta\in\Theta$,  $p_{\theta}^{l,M,N_{l}}(y_{1:T})$, $(u_1^l,\dots,u_T^l)\in\mathsf{E}_l^T$,  
$(\widetilde{u}_1^{l-1},\dots,\widetilde{u}_T^{l-1})\in\mathsf{E}_{l-1}^T$.
}
\item{Propose $\theta'|\theta$ from a proposal kernel $r_2(\theta,\theta')d\theta'$.}
\item{Run Algorithm \ref{alg:delta_pf} with the given $\theta'$ and other input parameters.}
\item{Generate $Z\sim\mathcal{U}_{[0,1]}$
and return $(\theta',p_{\theta'}^{l,M,N_{l}}(y_{1:T}))$ and $(u_1^{',l},\dots,u_T^{',l})$ 
$(\widetilde{u}_1^{',l-1},\dots,\widetilde{u}_T^{',l-1})$ 
if
$$
Z<\min\left\{
1, \frac{p_{\theta'}^{l,M,N_{l}}(y_{1:T})\nu(\theta')r_2(\theta',\theta)}
{
p_{\theta}^{l,M,N_{l}}(y_{1:T})\nu(\theta)
r_2(\theta,\theta')}
\right\}
$$
otherwise return $(\theta,p_{\theta}^{l,M,N_{l}}(y_{1:T}))$ and $(u_1^l,\dots,u_T^l)$, 
$(\widetilde{u}_1^{l-1},\dots,\widetilde{u}_T^{l-1})$.
}
\end{enumerate}
\caption{Bi-Level Particle Marginal Metropolis Kernel at level $l\in\mathbb{N}$.}
\label{alg:pmmh_l}
\end{algorithm}

\subsection{Final Algorithm}\label{sec:fa}

The way in which we proceed is then as follows.
\begin{enumerate}
\item{Run Algorithm \ref{alg:pmmh_0} at some level $l_{\star}\in\mathbb{N}$ for $I_{l_{\star}}$ iterations with $N_{l_{\star}}$
samples for Algorithm \ref{alg:basic_method}.}
\item{Independently of step 1.~and independently for each $l\in\{l_{\star}+1,\dots,L\}$ run Algorithm \ref{alg:pmmh_l} for $I_l$ iterations with $N_{l}$
samples for Algorithm \ref{alg:basic_method_coup}.}
\end{enumerate}
Let $\varphi:\Theta\times\mathbb{R}^{Td}\rightarrow\mathbb{R}$ and set
$$
\pi(\varphi) := \int_{\Theta\times\mathbb{R}^{Td}}\varphi(\theta,x_{1:T}) \pi\left(d(\theta,x_{1:T})|y_{1:T}\right)
$$
which we assume is finite.  Then the estimator of $\pi(\varphi)$ that we will use is
$$
\widehat{\pi(\varphi)} := \overline{\pi}^{l_{\star},N_{l_{\star}},I_{l_{\star}}}(\varphi) + \sum_{l=l_{\star}}^L\left\{
\overline{\pi}^{l,N_l,I_l}(\varphi) - \overline{\pi}^{l-1,N_{l},I_l}(\varphi)
\right\}
$$
where $\overline{\pi}^{l_{\star},N_{l_{\star}},I_{l_{\star}}}(\varphi)$ is as computed in \eqref{eq:mcmc_est}
and $\overline{\pi}^{l,N_l,I_l}(\varphi) - \overline{\pi}^{l-1,N_{l},I_l}(\varphi)$ is as described in \eqref{eq:mcmc_bl_est}.  We remark that there are several errors and biases including the time-discretization,  the use of
particle approaches in Algorithms \ref{alg:basic_method} and \ref{alg:basic_method_coup} as well as the MCMC bias.  We will show by a suitable selection of $L$,  $I_{l_{\star}},\dots,I_L$ and $N_{l_{\star}},\dots,N_L$ how these errors and biases can be controlled; this is the topic of the next section.

\section{Theory}\label{sec:theory}

We now start by giving the main mathematical result of the article.   The assumptions and technical results can be found in Appendix \ref{app:appendix}. $\mathbb{E}$ denotes the expectation w.r.t.~the process that is described in Section \ref{sec:fa}. Note that $T$ is suppressed in the notation,  but in practice the constant $C$ below would explode exponentially in $T$.

\begin{theorem}\label{theo:main_res}
Assume (A\ref{ass:1}-\ref{ass:3}). Then
for any  $\varphi\in\mathcal{C}_b^2(\Theta\times\mathbb{R}^{dT})\cap\mathcal{B}_b(\Theta\times\mathbb{R}^{dT})$ there exists a $C<+\infty$
such that for any $(l_{\star},L,N_{l_{\star}},I_{l_{\star}},\dots,N_L,I_L)\in\mathbb{N}^{2(L-l_{\star})+4}$ with $l_{\star}<L$
$$
\mathbb{E}\left[
\left(
\widehat{\pi(\varphi)} - \pi(\varphi)
\right)^2
\right]  \leq 
C
\left(\frac{1}{I_{l_{\star}}+1} +
\sum_{l=l_{\star}+1}^L\frac{\Delta_l}{I_l+1}
+ \sum_{(l,q)\in\mathsf{A}_{l_{\star},L}}\frac{\Delta_l^{1/2}}{I_l+1}\frac{\Delta_q^{1/2}}{I_q+1}
+ \frac{1}{N_{l_{\star}}}
+ \left(\sum_{l=l_{\star}+1}^L\frac{\Delta_l^{1/2}}{N_l^{1/2}}\right)^2 + \Delta_L^2
\right)
$$
where $\mathsf{A}_{l_{\star},L}=\{(l,q)\in\{l_{\star}+1,\dots,L\}:l\neq q\}$.
\end{theorem}

\begin{proof}
For $\varphi\in\mathcal{B}_b(\Theta\times\mathbb{R}^{Td})$,  we will write:
\begin{eqnarray*}
\overline{\pi}^l(\varphi) & := & \int_{\Theta\times\mathsf{E}_l^T}\varphi(\theta,x_{1:T})
\overline{\pi}^{l}\left(d(\theta,u_1,\dots,u_T)\right) \\
\overline{\pi}^{l,N_l}(\varphi) & := &
\int_{\Theta\times\mathsf{E}_l^T}\varphi(\theta,x_{1:T})
\overline{\pi}^{l,N_l}\left(d(\theta,u_1,\dots,u_T)\right).
\end{eqnarray*}
Using the $C_2-$inequality five times we have that
$$
\mathbb{E}\left[
\left(
\widehat{\pi(\varphi)} - \pi(\varphi)
\right)^2
\right]  \leq C\sum_{j=1}^5 T_j
$$
where $C=2^5$,  
\begin{eqnarray*}
T_1 & = & \mathbb{E}\left[\left(\overline{\pi}^{l_{\star},N_{l_{\star}},I_{l_{\star}}}(\varphi)-
\overline{\pi}^{l_{\star},N_{l_{\star}}}(\varphi)
\right)^2\right] \\
T_2 & = & 
\mathbb{E}\left[\left(
\sum_{l=l_{\star}}^L\left\{
\overline{\pi}^{l,N_l,I_l}(\varphi) - \overline{\pi}^{l-1,N_{l},I_l}(\varphi) - 
\{
\overline{\pi}^{l,N_l}(\varphi) - \overline{\pi}^{l-1,N_{l}}(\varphi)
\}
\right\}
\right)^2\right]\\
T_3 & = & \left(\overline{\pi}^{l_{\star},N_{l_{\star}}}(\varphi) - \overline{\pi}^{l_{\star}}(\varphi)\right)^2\\
T_4 & = & \left\{\sum_{l=l_{\star}+1}^L
\left\{
\overline{\pi}^{l,N_l}(\varphi) - 
\overline{\pi}^{l-1,N_l}(\varphi) -
\left\{\overline{\pi}^{l}(\varphi) - 
\overline{\pi}^{l-1}(\varphi)
\right\}
\right\}\right\}^2\\
T_5 & = & \left(\overline{\pi}^{L}(\varphi) - \overline{\pi}(\varphi)\right)^2.
\end{eqnarray*}
$T_1$ and $T_5$ can be controlled using standard results for uniformly ergodic Markov chains and weak errors of Euler-Maruyama time discretizations; the bounds are $\mathcal{O}((I_{l_{\star}}+1)^{-1})$ and $\mathcal{O}(\Delta_L^2)$ respectively.  The term $T_2$ is controlled in Lemmata \ref{lem:lem4}-\ref{lem:lem5} in the Appendix,  $T_3$ is bounded in Lemma \ref{lem:lem2} and $T_4$ in Lemma \ref{lem:lem1}.  This completes the proof of the theorem.
\end{proof}

The consequence of this result is follows. Let $\epsilon>0$ be given.  Supposing that $M=\mathcal{O}(T)$ in Algorithms \ref{alg:cond_pf_0}
and \ref{alg:delta_pf} as is often the case in practice (e.g.~\cite{andrieu}) and then ignoring the cost associated to the total observation time $T$,  
one has a cost of the final algorithm (as in Section \ref{sec:fa}) of $\mathcal{O}\left(\sum_{l=l_{\star}}^LI_l\Delta_l^{-1}N_l^2\right)$.
If one supposes the cost of the level $l_{\star}$ is neglibible (which is an underlying principle of the MLMC framework) one can set $I_l=\mathcal{O}(\epsilon^{-2}\Delta_l^{6/7})$,  $N_l=\mathcal{O}(\epsilon^{-2}\Delta_l^{1/2})$ and $L=\mathcal{O}(|\log(\epsilon)|)$ which would give an upper-bound on the MSE
(i.e.~$\mathbb{E}\left[\left(\widehat{\pi(\varphi)} - \pi(\varphi)\right)^2\right]$) of 
$\mathcal{O}(\epsilon^2)$ for a cost of 
$\mathcal{O}(\epsilon^{-6})$.
If one simply used one level $L$ along with Algorithm \ref{alg:pmmh_0},  then to achieve an MSE of 
$\mathcal{O}(\epsilon^2)$ it would have a cost of 
$\mathcal{O}(\epsilon^{-7})$.  That is to say that our using a multilevel method reduces the cost by an order of magnitude.

\section{Numerical Simulations}\label{sec:numerics}

\subsection{Models}
\subsubsection{Kuramoto Model}
We examine the Kuramoto model \(\{X_t\}_{t \geq 0}\) defined by the McKean-Vlasov stochastic differential equation:
\[
dX_t = \left( \theta + \int \sin(X_t - y) \, d\mu_t(y) \right) dt + \sigma dW_t, \quad X_0 = x_0 \in \mathbb{R}, \; t \in [0, T],
\]
where \(\mu_t\) represents the distribution of \(X_t\), \(\theta\) is a random variable, \(\sigma > 0\), and \(T \in \mathbb{N}\). The Kuramoto model is a prominent mean field game model that is often utilized for numerical simulations.

\subsubsection{Modified Kuramoto Model}

We also consider the process \(\{ X_t \}_{t \geq 0}\) defined by the McKean-Vlasov stochastic differential equation:
\[
dX_t = \left( \theta + \int \sin(X_t - y) \, d\mu_t(y) \right) dt + \left( \frac{\sigma}{1 + X_t^2} \right) dW_t, \quad X_0 = x_0 \in \mathbb{R}, \; t \in [0, T],
\]
where \( \mu_t \) is the distribution of \( X_t \), \( \theta \) is a random variable, \( \sigma > 0 \), and \( T \in \mathbb{N} \). We refer to this as the modified Kuramoto model due to the use of a non-constant diffusion coefficient.\\

\subsection{Simulation Results}

We conduct simulations for both models with observations \(\{Y_k\}_{k=1}^{T}\) at unit time intervals, where \( Y_k | X_k = x_k \sim \mathcal{N}(x_k, \tau^2) \) for \( k \in \{1, \ldots, T\} \) with the starting point  \( x_0 = 1 \), and \( T = 100 \).  Here $\mathcal{N}(x,\tau)$ is a normal distribution of mean $x$ and variance $\tau$.
The true values of the unknown parameters are set to be \( \theta = 0 \), \( \sigma = 0.2 \) and \( \tau = 1 \).  

Independent Gaussian priors are assumed for \(\theta\), \(\log(\sigma)\), and \(\log(\tau)\). Gaussian random walk proposals with a diagonal covariance matrix are used for MCMC moves, with the sampler tuned to ensure good mixing.  We choose the simulation parameters as in Section \ref{sec:theory}.


Figures \ref{fig:conv_KU} and \ref{fig:conv_MKU} illustrate the convergence performance of the PMCMC and MLPMCMC for the Kuramoto Model and the Modified Kuramoto Model, respectively, considering discretization levels up to \(l = 5\). With $5000$ samples, both models exhibit generally good mixing across the two samplers.

We examined the logarithmic relationship between computational cost and mean squared error (MSE) reduction for both models. Tables \ref{tab:rate_KU} and \ref{tab:rate_MKU} present the estimated rates, revealing values close to \(-3\) for the multilevel method (indicating the cost of $\mathcal{O}(\epsilon^{-6})$)
and \(-3.5\) for the single-level method (indicating the cost of $\mathcal{O}(\epsilon^{-7})$) for each parameter. The observed rates are consistent with theoretical expectations, although some deviation is attributable to errors in distribution estimation.

\begin{figure}[htbp]
    \centering
    \includegraphics[scale=0.35, trim={2.5cm 3cm 2.5cm 3cm},clip]{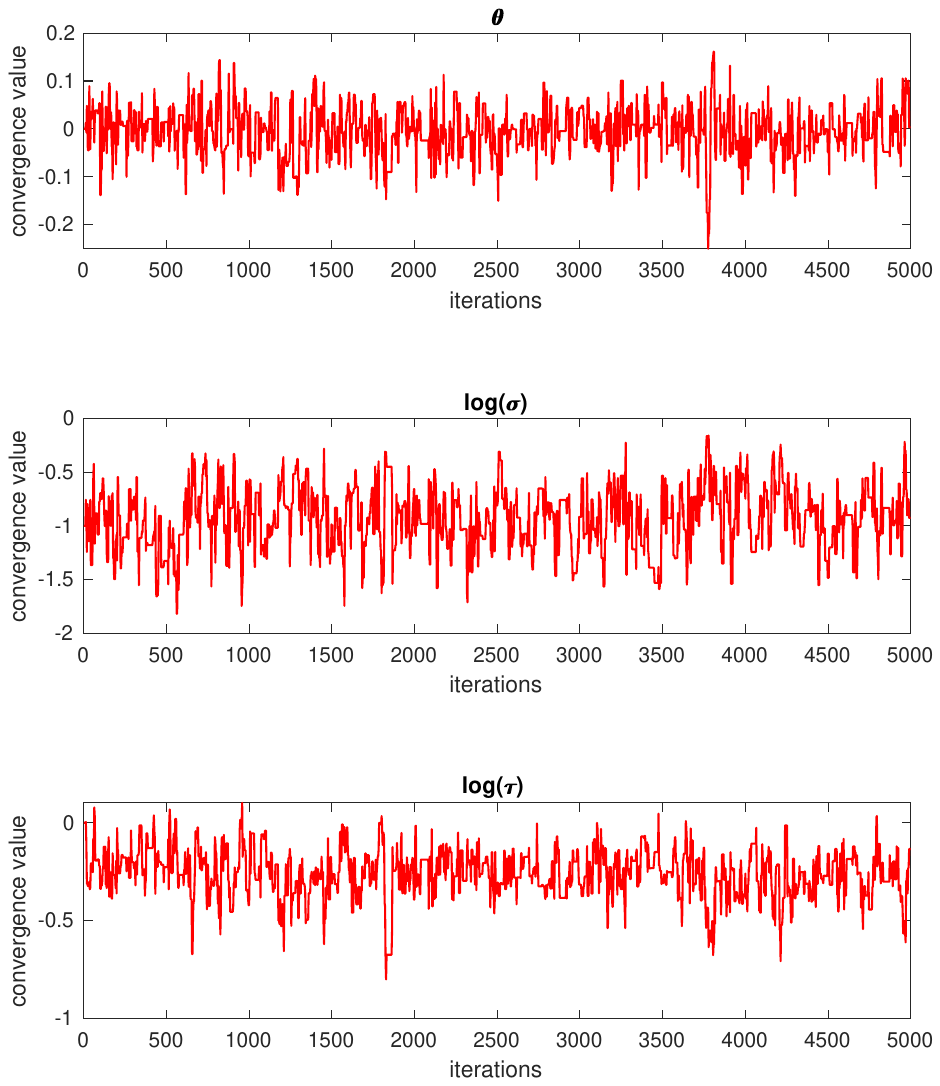}
    \includegraphics[scale=0.35, trim={2.5cm 3cm 2.5cm 3cm},clip]{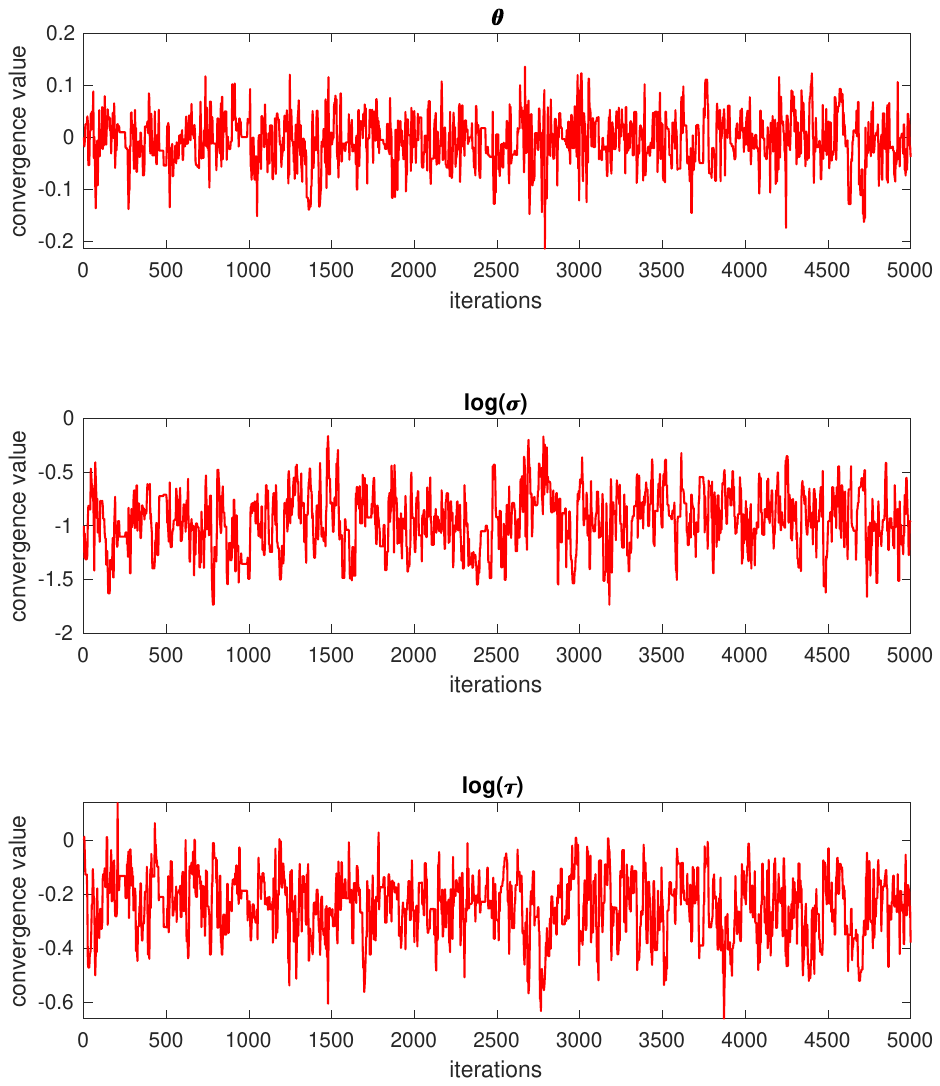}
    \caption{The convergence plot of the PMCMC (left) and MLPMCMC (right) for the Kuramoto model.}
    \label{fig:conv_KU}
\end{figure}

\begin{figure}[htbp]
    \centering
    \includegraphics[scale=0.35, trim={2.5cm 3cm 2.5cm 3cm},clip]{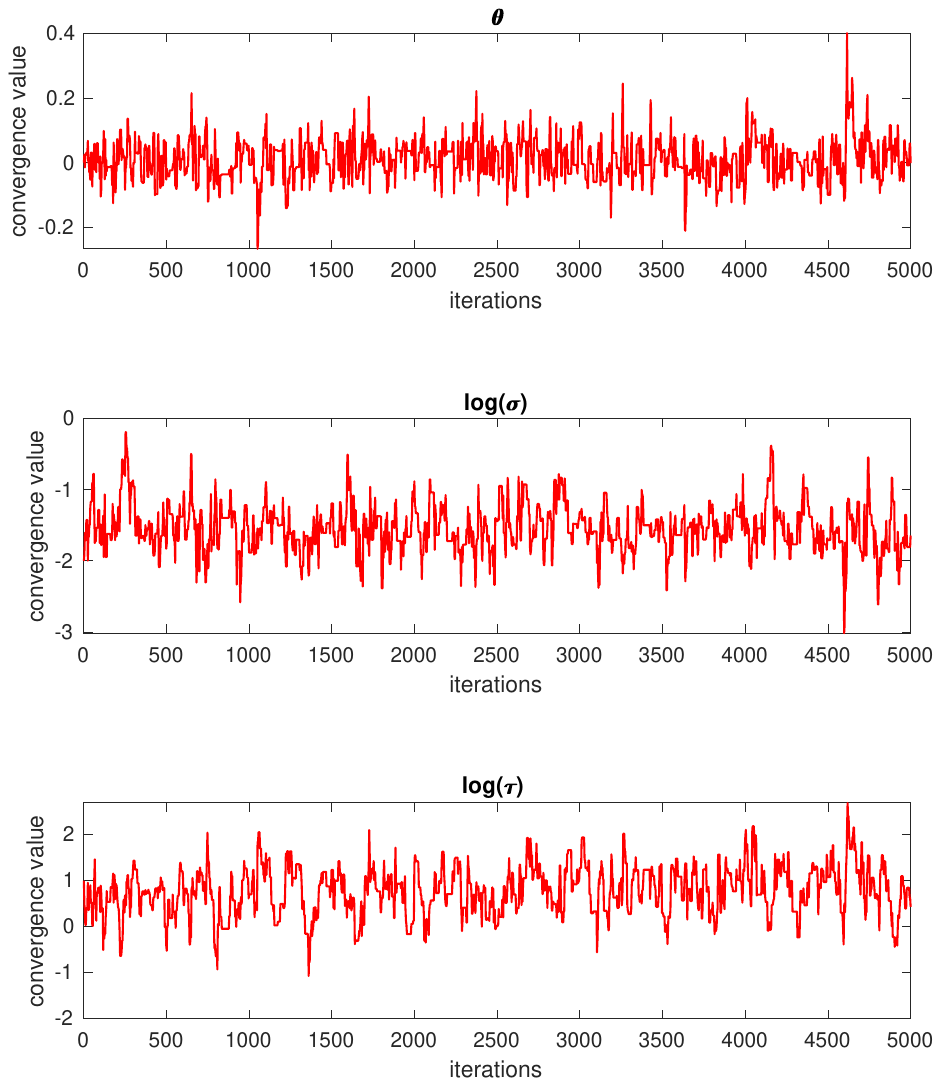}
    \includegraphics[scale=0.35, trim={2.5cm 3cm 2.5cm 3cm},clip]{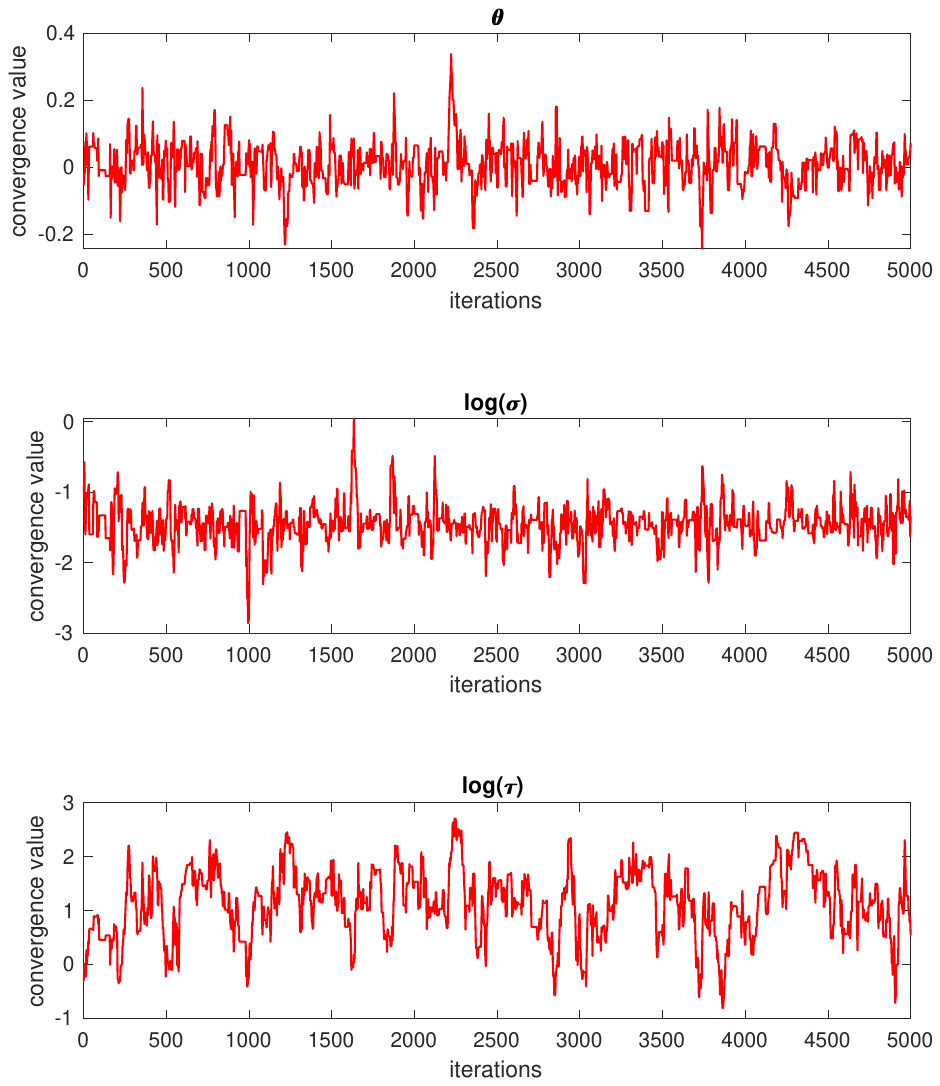}
    \caption{The convergence plot of the PMCMC (left) and MLPMCMC (right) for the Modified Kuramoto model.}
    \label{fig:conv_MKU}
\end{figure}

\begin{table}[h!]
    \centering
    \caption{Kuramoto Model: Estimated Log Cost against Log MSE.}
    \begin{tabular}{c|c|c}
    Parameter & PMCMC & MLPMCMC \\
    \hline
    $\theta$ & -3.33 & -2.71 \\
    $\log(\sigma)$ & -3.69 & -3.25 \\
    $\log(\tau)$ & -3.60 & -3.17\\
    \end{tabular}
    \label{tab:rate_KU}
\end{table}

\begin{table}[h!]
    \centering
    \caption{Modified Kuramoto Model: Estimated Log Cost against Log MSE.}
    \begin{tabular}{c|c|c}
    Parameter & PMCMC & MLPMCMC \\
    \hline
    $\theta$ & -3.52 & -2.94 \\
    $\log(\sigma)$ & -3.56 & -2.75 \\
    $\log(\tau)$ & -3.85 & -2.90\\
    \end{tabular}
    \label{tab:rate_MKU}
\end{table}

\subsubsection*{Acknowledgements}

AJ was supported by CUHK-SZ start-up funding.

\appendix

\section{Technical Proofs}\label{app:appendix}

We make the following assumptions.

\begin{hypA}\label{ass:1}
\begin{enumerate}
\item{For each $(\theta,\mu)\in\Theta\times\mathcal{P}(\mathbb{R}^d)$, $(a_{\theta}\left(\cdot,\overline{\xi}_{\theta}(\cdot,\mu)\right)
,\xi_{\theta}(\cdot))
\in\mathcal{C}_b^2(\mathbb{R}^{d+1},\mathbb{R}^d)\cap\mathcal{B}_b(\mathbb{R}^{d+1},\mathbb{R}^d)
\times\mathcal{C}_b^2(\mathbb{R}^{2d},\mathbb{R})\cap\mathcal{B}_b(\mathbb{R}^{2d},\mathbb{R})$.
All upper-bounds of the functionals are true for every $\theta\in\Theta$.
}
\item{$\sigma(\cdot)\in\mathcal{C}_b^2(\mathbb{R}^{d},\mathbb{R}^{d\times d})\cap
\mathcal{B}_b(\mathbb{R}^{d},\mathbb{R}^{d\times d})$.}
\item{Set $\Sigma(x)=\sigma(x)\sigma(x)^{\top}$, then for any $x\in\mathbb{R}^d$, $\Sigma(x)$ is positive definite.}
\end{enumerate}
\end{hypA}

\begin{hypA}\label{ass:2}
For every $y\in \mathsf{Y}$ the function $(\theta,x)\mapsto G_{\theta}(x,y)$ belongs to $\mathcal{C}_b^2(\Theta\times\mathbb{R}^d)\cap\mathcal{B}_b(\Theta\times\mathbb{R}^d)$. Furthermore we assume that $\inf_{(\theta,x,y)\in\Theta\times\mathbb{R}^{d}\times\mathsf{Y}}G_{\theta}(x,y)>0$.
\end{hypA}

\begin{hypA}\label{ass:3}
For each $l\in\mathbb{N}$ the Markov kernels $K_l$ and $\check{K}_l$ admit an invariant measure for which
they are reversible with respect to.  Moreover they are uniformly ergodic with 1-step minorization conditions that are uniform in $l$.
\end{hypA}

The assumptions above are made for the following reasons.  (A\ref{ass:1}) has been made in \cite{po_mv} except when there is no parameter. We make the assumption as many of our proofs rely on the work in \cite{po_mv}. The constants are made uniform in $\theta$ as this allows us to pass seemlessly from the results in \cite{po_mv} to this paper.  They would often only occur if $\Theta$ is a compact space.   (A\ref{ass:2}-\ref{ass:3}) are used in \cite{jasra_bpe_sde} as again our proof strategy uses some of the results of the afore-mentioned article;  we refer the reader there for further discussion.

Throughout $C$ is a finite constant whose value may change upon appearance.  Any dependencies on simulation or model parameters will be clear from the statement of each result.  We will now give a series of technical results which should be read in order and are used to prove Theorem \ref{theo:main_res}.

\begin{lem}\label{lem:lem1}
Assume (A\ref{ass:1}-\ref{ass:2}).  Then
for any  $\varphi\in\mathcal{C}_b^2(\Theta\times\mathbb{R}^{dT})\cap\mathcal{B}_b(\Theta\times\mathbb{R}^{dT})$ there exists a $C<+\infty$
such that for any $(l_{\star},L,N_{l_{\star}},\dots,N_L)\in\mathbb{N}^{L-l_{\star}+3}$ with $l_{\star}<L$
$$
\left|\sum_{l=l_{\star}+1}^L
\left(
\overline{\pi}^{l,N_l}(\varphi) - 
\overline{\pi}^{l-1,N_l}(\varphi) -
\left\{\overline{\pi}^{l}(\varphi) - 
\overline{\pi}^{l-1}(\varphi)
\right\}
\right)
\right|
\leq C\sum_{l=l_{\star}+1}^L\frac{\Delta_l^{1/2}}{N_l^{1/2}}.
$$
\end{lem}

\begin{proof}
We can write
\begin{eqnarray}
\overline{\pi}^l(\varphi) & = & \frac{\int_{\Theta}\overline{P}_{T,\theta}^l\left(
\varphi\prod_{k=1}^TG_{k,\theta}
\right)\nu(\theta)d\theta}{\int_{\Theta}\overline{P}_{T,\theta}^l\left(\prod_{k=1}^TG_{k,\theta}
\right)\nu(\theta)d\theta} \nonumber\\
\overline{\pi}^{l,N_l}(\varphi) & = & \frac{\int_{\Theta}
\overline{\mathbb{E}}_{\theta}^{N_l}\left[
\overline{P}_{T,\theta}^{l,N_l}\left(
\varphi\prod_{k=1}^TG_{k,\theta}
\right)\right]\nu(\theta)d\theta}{\int_{\Theta}
\overline{\mathbb{E}}_{\theta}^{N_l}\left[
\overline{P}_{T,\theta}^{l,N_l}\left(\prod_{k=1}^TG_{k,\theta}
\right)\nu(\theta)\right]d\theta}\label{eq:stoch_post}
\end{eqnarray}
where
\begin{eqnarray*}
\overline{P}_{T,\theta}^l\left(
\varphi\prod_{k=1}^TG_{k,\theta}
\right) & := &
\int_{\Theta\times\mathsf{E}_l^T}\varphi(\theta,x_{1:T})
\left\{\prod_{k=1}^T G_{\theta}(x_k,y_k)
\right\}\prod_{k=1}^T \overline{P}_{\mu_{k-1,\theta}^l,k,\theta}^l(x_{k-1},du_k) \\
\overline{P}_{T,\theta}^{l,N_l}\left(
\varphi\prod_{k=1}^TG_{k,\theta}
\right) & := &
\int_{\Theta\times\mathsf{E}_l^T}\varphi(\theta,x_{1:T})
\left\{\prod_{k=1}^T G_{\theta}(x_k,y_k)\right\}
\prod_{k=1}^T \overline{P}_{\mu_{k-1,\theta}^{l,N_l},k,\theta}^l(x_{k-1},du_k).
\end{eqnarray*}
Therefore we have
$$
\left|\sum_{l=l_{\star}+1}^L
\left(
\overline{\pi}^{l,N_l}(\varphi) - 
\overline{\pi}^{l-1,N_l}(\varphi) -
\left\{\overline{\pi}^{l}(\varphi) - 
\overline{\pi}^{l-1}(\varphi)
\right\}
\right)
\right|
\leq 
$$
$$
\sum_{l=l_{\star}+1}^L
\Bigg|
\frac{\int_{\Theta}
\overline{\mathbb{E}}_{\theta}^{N_l}\left[
\overline{P}_{T,\theta}^{l,N_l}\left(
\varphi\prod_{k=1}^TG_{k,\theta}
\right)\right]\nu(\theta)d\theta}{\int_{\Theta}
\overline{\mathbb{E}}_{\theta}^{N_l}\left[
\overline{P}_{T,\theta}^{l,N_l}\left(\prod_{k=1}^TG_{k,\theta}
\right)\right]\nu(\theta)d\theta} -
\frac{\int_{\Theta}
\overline{\mathbb{E}}_{\theta}^{N_l}\left[
\overline{P}_{T,\theta}^{l-1,N_l}\left(
\varphi\prod_{k=1}^TG_{k,\theta}
\right)\right]\nu(\theta)d\theta}{\int_{\Theta}
\overline{\mathbb{E}}_{\theta}^{N_l}\left[
\overline{P}_{T,\theta}^{l-1,N_l}\left(\prod_{k=1}^TG_{k,\theta}
\right)\right]\nu(\theta)d\theta}
-
$$
$$
\left\{ \frac{\int_{\Theta}\overline{P}_{T,\theta}^l\left(
\varphi\prod_{k=1}^TG_{k,\theta}
\right)\nu(\theta)d\theta}{\int_{\Theta}\overline{P}_{T,\theta}^l\left(\prod_{k=1}^TG_{k,\theta}
\right)\nu(\theta)d\theta} -
\frac{\int_{\Theta}\overline{P}_{T,\theta}^{l-1}\left(
\varphi\prod_{k=1}^TG_{k,\theta}
\right)\nu(\theta)d\theta}{\int_{\Theta}\overline{P}_{T,\theta}^{l-1}\left(\prod_{k=1}^TG_{k,\theta}
\right)\nu(\theta)d\theta}
\right\}
\Bigg|.
$$
We can consider each summand on the right hand side (R.H.S.) individually.
Recall the identity for real numbers $A,B,C,D,a,b,c,d$ with $A,B,C,D$ non-zero:
\begin{equation}
\begin{split}
    \frac{a}{A} - \frac{b}{B} - \frac{c}{C} + \frac{d}{D} =\,& \frac{1}{A}(a-b-c+d) -\frac{b}{AB}(A-B-C+D) - \frac{1}{AC}(A-C)(c-d)\nonumber\\
    -\,&\frac{1}{AB}(C-D)(b-d)+\frac{d}{CBD}(B-D)(C-D) + \frac{d}{ACB}(A-C)(C-D). \label{eq:diff_ratio}\\
\end{split}
\end{equation}
To continue with the proof,  we consider only terms of the type
$$
 \frac{1}{A}(a-b-c+d) \quad\textrm{and}\quad \frac{1}{AC}(A-C)(c-d).
$$
All the other terms can be dealt with using a similar approach.

Then we have to consider
$$
\Bigg(\int_{\Theta}
\overline{\mathbb{E}}_{\theta}^{N_l}\left[
\overline{P}_{T,\theta}^{l,N_l}\left(\prod_{k=1}^TG_{k,\theta}
\right)\nu(\theta)\right]d\theta\Bigg)^{-1}\times
$$
$$
\Bigg\{\int_{\Theta}
\overline{\mathbb{E}}_{\theta}^{N_l}\left[
\overline{P}_{T,\theta}^{l,N_l}\left(
\varphi\prod_{k=1}^TG_{k,\theta}
\right)\right]\nu(\theta)d\theta -
\int_{\Theta}
\overline{\mathbb{E}}_{\theta}^{N_l}\left[
\overline{P}_{T,\theta}^{l-1,N_l}\left(
\varphi\prod_{k=1}^TG_{k,\theta}
\right)\right]\nu(\theta)d\theta -
$$
\begin{equation}\label{eq:prf1}
\left(
\int_{\Theta}\overline{P}_{T,\theta}^l\left(
\varphi\prod_{k=1}^TG_{k,\theta}
\right)\nu(\theta)d\theta -
\int_{\Theta}\overline{P}_{T,\theta}^{l-1}\left(
\varphi\prod_{k=1}^TG_{k,\theta}
\right)\nu(\theta)d\theta
\right)
\Bigg\}.
\end{equation}
Now consider the situation where we run four processes.  The first two are 
synchronously coupled Euler-Maruyama discretizations
of \eqref{eq:sde} at levels $l$ and $l-1$ where the laws have been approximated using
Algorithm \ref{alg:basic_method_coup} ($N_l$ samples) and the other two processes are 
synchronously coupled
 exact Euler-Maruyama discretizations
of \eqref{eq:sde} at levels $l$ and $l-1$, critically where the Brownian motions are all the same for all four processes.   We write the processes at time $t$ as $(X_t^{l,N_l},X_t^{l-1,N_l},X_t^{l},X_t^{l-1})$ on writing
expectations w.r.t.~the afore-mentioned process as $\widetilde{\mathbb{E}}_{\theta}^{l,N_l}$ we have that the numerator
of \eqref{eq:prf1} can be rewritten as
$$
\int_{\Theta}
\widetilde{\mathbb{E}}_{\theta}^{l,N_l}\Bigg[
\varphi(\theta,X_{1:T}^{l,N_l})\prod_{k=1}^T G_{k}(X_k^{l,N_l},y_k) - 
\varphi(\theta,X_{1:T}^{l-1,N_l})\prod_{k=1}^T G_{k}(X_k^{l-1,N_l},y_k) - 
$$
\begin{equation}\label{eq:prf2}
\left\{
\varphi(\theta,X_{1:T}^{l})\prod_{k=1}^T G_{k}(X_k^{l},y_k) - 
\varphi(\theta,X_{1:T}^{l-1})\prod_{k=1}^T G_{k}(X_k^{l-1},y_k)
\right\}
\Bigg]
\nu(\theta)d\theta
\end{equation}
Then we can use the fact that the function $\varphi(\theta,x_{1:T})\prod_{k=1}^T G_{k}(x_k,y_k)$ is in the
collection of functions that are $\mathcal{C}_b^2(\Theta\times\mathbb{R}^{dT})\cap\mathcal{B}_b(\Theta\times\mathbb{R}^{dT})$
to yield that \eqref{eq:prf2} is upper-bounded by (see e.g.~\cite[Lemma A.5.]{po_mv})
$$
C\left(\int_{\Theta}
\widetilde{\mathbb{E}}_{\theta}^{l,N_l}\Bigg[\|
(X_{1:T}^{l,N_l} - X_{1:T}^{l-1,N_l}) - 
(X_{1:T}^{l} - X_{1:T}^{l-1})
\| + \|X_{1:T}^{l} - X_{1:T}^{l-1}\|\{
\|X_{1:T}^{l,N_l}-X_{1:T}^{l}\| + 
\|X_{1:T}^{l-1,N_l}-X_{1:T}^{l-1}\|\}
\bigg]
\nu(\theta)d\theta
\right).
$$
Then in conjunction with the Cauchy-Schwarz inequality,  one can use \cite[Lemma A.7]{po_mv} 
and (A\ref{ass:2}) to deduce that
\eqref{eq:prf1} is upper-bounded by
$$
C\frac{\Delta_l^{1/2}}{N_l^{1/2}}.
$$

The next term that we consider is
$$
\left(
\left\{\int_{\Theta}
\overline{\mathbb{E}}_{\theta}^{N_l}\left[
\overline{P}_{T,\theta}^{l,N_l}\left(\prod_{k=1}^TG_{k,\theta}
\right)\right]\nu(\theta)d\theta\right\}\left\{
\int_{\Theta}\overline{P}_{T,\theta}^l\left(\prod_{k=1}^TG_{k,\theta}
\right)\nu(\theta)d\theta\right\}
\right)^{-1}\times
$$
$$
\left(
\int_{\Theta}
\overline{\mathbb{E}}_{\theta}^{N_l}\left[
\overline{P}_{T,\theta}^{l,N_l}\left(\prod_{k=1}^TG_{k,\theta}
\right)\right]\nu(\theta)d\theta - 
\int_{\Theta}\overline{P}_{T,\theta}^l\left(\prod_{k=1}^TG_{k,\theta}
\right)\nu(\theta)d\theta
\right)\times
$$
\begin{equation}\label{eq:prf3}
\left(
\int_{\Theta}\overline{P}_{T,\theta}^l\left(
\varphi\prod_{k=1}^TG_{k,\theta}
\right)\nu(\theta)d\theta -
\int_{\Theta}\overline{P}_{T,\theta}^{l-1}\left(
\varphi\prod_{k=1}^TG_{k,\theta}
\right)\nu(\theta)d\theta
\right).
\end{equation}
Using (A\ref{ass:2}) the denominator of \eqref{eq:prf3} is lower-bounded,  so on using the representation
of the processes $(X_t^{l,N_l},X_t^{l-1,N_l},$ $X_t^{l},X_t^{l-1})$ as above,  we have that
 \eqref{eq:prf3}  is upper-bounded by
$$
C
\int_{\Theta}\left|
\widetilde{\mathbb{E}}_{\theta}^{l,N_l}\Bigg[
\prod_{k=1}^T G_{k}(X_k^{l,N_l},y_k) - 
\prod_{k=1}^T G_{k}(X_k^{l},y_k)\Bigg]\right|\nu(\theta)d\theta\times
$$
$$
\Bigg|
\int_{\Theta}\overline{P}_{T,\theta}^l\left(
\varphi\prod_{k=1}^TG_{k,\theta}
\right)\nu(\theta)d\theta -
\int_{\Theta}\overline{P}_{T,\theta}^{l-1}\left(
\varphi\prod_{k=1}^TG_{k,\theta}
\right)\nu(\theta)d\theta
\Bigg|.
$$
The first term can be controlled using
\cite[Lemma A.7]{po_mv} and the second by using well-known results for weak errors of diffusion
models to give that \eqref{eq:prf3}  is upper-bounded by
$$
C\frac{\Delta_l}{N_l^{1/2}}.
$$
As noted the other terms can be handeled in a similar manner,  which concludes the proof.
\end{proof}

\begin{lem}\label{lem:lem2}
Assume (A\ref{ass:1}-\ref{ass:2}).  Then
for any  $\varphi\in\mathcal{C}_b^2(\Theta\times\mathbb{R}^{dT})\cap\mathcal{B}_b(\Theta\times\mathbb{R}^{dT})$ there exists a $C<+\infty$
such that for any $(l,N_{l})\in\mathbb{N}^{2}$ 
$$
\left|
\overline{\pi}^{l,N_l}(\varphi) - 
\overline{\pi}^{l}(\varphi) 
\right|
\leq \frac{C}{N_l^{1/2}}.
$$
\end{lem}

\begin{proof}
Using the notation of the proof of Lemma \ref{lem:lem1} we have that 
$$
\overline{\pi}^{l,N_l}(\varphi) - 
\overline{\pi}^{l}(\varphi) 
 = 
$$
$$
\frac{\int_{\Theta}
\overline{\mathbb{E}}_{\theta}^{N_l}\left[
\overline{P}_{T,\theta}^{l,N_l}\left(
\varphi\prod_{k=1}^TG_{k,\theta}
\right)\right]\nu(\theta)d\theta}{\int_{\Theta}
\overline{\mathbb{E}}_{\theta}^{N_l}\left[
\overline{P}_{T,\theta}^{l,N_l}\left(\prod_{k=1}^TG_{k,\theta}
\right)\nu(\theta)\right]d\theta \int_{\Theta}\overline{P}_{T,\theta}^l\left(\prod_{k=1}^TG_{k,\theta}
\right)\nu(\theta)d\theta}\times
$$
$$
\left(
\int_{\Theta}\overline{P}_{T,\theta}^l\left(\prod_{k=1}^TG_{k,\theta}
\right)\nu(\theta)d\theta
-
\int_{\Theta}
\overline{\mathbb{E}}_{\theta}^{N_l}\left[
\overline{P}_{T,\theta}^{l,N_l}\left(\prod_{k=1}^TG_{k,\theta}
\right)\nu(\theta)\right]d\theta 
\right) + 
$$
\begin{equation}\label{eq:prf4}
\frac{1}{\int_{\Theta}\overline{P}_{T,\theta}^l\left(\prod_{k=1}^TG_{k,\theta}
\right)\nu(\theta)d\theta
}
\left(
\int_{\Theta}
\overline{\mathbb{E}}_{\theta}^{N_l}\left[
\overline{P}_{T,\theta}^{l,N_l}\left(
\varphi\prod_{k=1}^TG_{k,\theta}
\right)\right]\nu(\theta)d\theta - 
\int_{\Theta}\overline{P}_{T,\theta}^l\left(\varphi\prod_{k=1}^TG_{k,\theta}
\right)\nu(\theta)d\theta
\right).
\end{equation}
The two terms on the R.H.S.~of \eqref{eq:prf4} can be controlled in almost identical ways,  so we just consider the first one.  By using $\varphi$ being upper-bounded and (A\ref{ass:2}) this afore-mentioned term is upper-bounded by 
$$
C\Bigg|
\int_{\Theta}\overline{P}_{T,\theta}^l\left(\prod_{k=1}^TG_{k,\theta}
\right)\nu(\theta)d\theta
-
\int_{\Theta}
\overline{\mathbb{E}}_{\theta}^{N_l}\left[
\overline{P}_{T,\theta}^{l,N_l}\left(\prod_{k=1}^TG_{k,\theta}
\right)\nu(\theta)\right]d\theta 
\Bigg|.
$$
Using the representation
of the processes $(X_t^{l,N_l},X_t^{l-1,N_l},X_t^{l},X_t^{l-1})$ as in the proof of Lemma \ref{lem:lem1} and 
\cite[Lemma A.7]{po_mv} we obtain the upper-bound  $C/N_l^{1/2}$; this suffices to complete the proof.
\end{proof}

\begin{lem}\label{lem:lem3}
Assume (A\ref{ass:1}-\ref{ass:2}).  Then
for any  $\varphi\in\mathcal{C}_b^2(\Theta\times\mathbb{R}^{dT})\cap\mathcal{B}_b(\Theta\times\mathbb{R}^{dT})$ there exists a $C<+\infty$ such that for any $(l,N_l)\in\mathbb{N}^2$
$$
\int_{\Theta\times(\mathsf{E}_l\times\mathsf{E}_{l-1})^T}
\left(\varphi(\theta,x_{1:T}^l)
\left\{\prod_{k=1}^T \check{H}_{k,\theta}(x_k^l,\widetilde{x}_{k}^{l-1})\right\} - 
\varphi(\theta,\widetilde{x}_{1:T}^l)
\left\{\prod_{k=1}^T \check{H}_{k,\theta}(\widetilde{x}_{k}^{l-1},x_k^l)\right\}
\right)^2
\check{\pi}^{l,N_{l}}\left(d(\theta,u_{1:T}^l,\widetilde{u}_{1:T}^{l-1})\right) \leq
$$
$$
C\Delta_l.
$$
\end{lem}

\begin{proof}
The proof is essentially identical to that of \cite[Lemma A.2.]{jasra_bpe_sde} with the only real difference is that one must use the strong error results that were developed in \cite[Lemma A.7]{po_mv} for the particle driven Euler-Maruyama discretizations that we use in this article.
\end{proof}

Recall that $\mathbb{E}$ denotes the expectation w.r.t.~the process that is described in Section \ref{sec:fa}.

\begin{lem}\label{lem:lem4}
Assume (A\ref{ass:1}-\ref{ass:3}).  Then
for any  $\varphi\in\mathcal{C}_b^2(\Theta\times\mathbb{R}^{dT})\cap\mathcal{B}_b(\Theta\times\mathbb{R}^{dT})$ there exists a $C<+\infty$
such that for any $(l_{\star},L,N_{l_{\star}},I_{l_{\star}},\dots,N_L,I_L)\in\mathbb{N}^{2(L-l_{\star})+4}$ with $l_{\star}<L$
$$
\sum_{l=l_{\star}+1}^L
\mathbb{E}\left[
\left(
\overline{\pi}^{l,N_l,I_l}(\varphi) - \overline{\pi}^{l-1,N_{l},I_l}(\varphi) -
\left\{\overline{\pi}^{l,N_l}(\varphi) - 
\overline{\pi}^{l-1,N_l}(\varphi) \right\}
\right)^2
\right]
\leq C\sum_{l=l_{\star}+1}^L\frac{\Delta_l}{I_l+1}.
$$
\end{lem}

\begin{proof}
This can be proved by using the representation \eqref{eq:stoch_post} and the identity \eqref{eq:diff_ratio}.
By using the $C_2-$inequality 6 times then one has 6 quantities to bound.  The difference of difference formulae can be treated using  \cite[Propositon A.1.]{jasra_bpe_sde} and Lemma \ref{lem:lem3}.  The other terms can be controlled using standard results for uniformly ergodic Markov chains and the calculations in the proof of Lemma \ref{lem:lem1}.  As the computations are more-or-less repetitive,  they are omitted for brevity.
\end{proof}

\begin{lem}\label{lem:lem5}
Assume (A\ref{ass:1}-\ref{ass:3}).  Then
for any  $\varphi\in\mathcal{C}_b^2(\Theta\times\mathbb{R}^{dT})\cap\mathcal{B}_b(\Theta\times\mathbb{R}^{dT})$ there exists a $C<+\infty$
such that for any $(l_{\star},L,N_{l_{\star}},I_{l_{\star}},\dots,N_L,I_L)\in\mathbb{N}^{2(L-l_{\star})+4}$ with $l_{\star}<L$
$$
\Bigg|\sum_{(l,q)\in\mathsf{A}_{l_{\star},L}}
\mathbb{E}\left[
\overline{\pi}^{l,N_l,I_l}(\varphi) - \overline{\pi}^{l-1,N_{l},I_l}(\varphi) -
\left\{\overline{\pi}^{l,N_l}(\varphi) - 
\overline{\pi}^{l-1,N_l}(\varphi) \right\}
\right]\times
$$
$$
\mathbb{E}\left[
\overline{\pi}^{q,N_q,I_q}(\varphi) - \overline{\pi}^{q-1,N_{q},I_q}(\varphi) -
\left\{\overline{\pi}^{q,N_q}(\varphi) - 
\overline{\pi}^{q-1,N_q}(\varphi) \right\}
\right]\Bigg|
\leq 
C\sum_{(l,q)\in\mathsf{A}_{l_{\star},L}}\frac{\Delta_l^{1/2}}{I_l+1}\frac{\Delta_q^{1/2}}{I_q+1}
$$
where $\mathsf{A}_{l_{\star},L}=\{(l,q)\in\{l_{\star}+1,\dots,L\}:l\neq q\}$.
\end{lem}

\begin{proof}
The proof is essentially the same as \cite[Proposition A.1]{jasra_cont} except with reference to already established results in this article.  As the calculations are repeated across the articles and this work,  we omit them for brevity.
\end{proof}

\end{document}